\documentclass[conference]{IEEEtran}
\usepackage{graphicx}
\usepackage{amsmath,amsthm}
\usepackage{amssymb}
\usepackage[vlined,linesnumbered]{algorithm2e}
\usepackage{array}
\usepackage{stfloats}
\usepackage[font=small]{caption}
\usepackage{url}
\usepackage{stmaryrd}
\usepackage{float}
\usepackage[backend=biber]{biblatex}
\usepackage{subcaption}
\usepackage{mathtools}
\usepackage{tikz}
\usetikzlibrary{shapes.gates.logic.US,trees,positioning,arrows,shapes.geometric}
\usetikzlibrary{calc}
\usepackage{wrapfig}
\usepackage{cutwin}
\usepackage{pgfplots}
\newcommand{\recht}[1]{\operatorname{#1}}
\newcommand{\End}{\recht{End}}
\newcommand{\OEnd}{\underline{\End}}
\newcommand{\CC}{\mathcal{C}}

\newcommand{\ob}{\recht{obj}}
\newcommand\iso{\xrightarrow{
   \,\smash{\raisebox{-0.65ex}{\ensuremath{\scriptstyle\sim}}}\,}}
\newcommand{\Hom}{\recht{Hom}}
\newcommand{\id}{\recht{id}}
\newcommand{\cat}[1]{\mathsf{#1}}
\newcommand{\tOR}{\mathtt{OR}}
\newcommand{\tAND}{\mathtt{AND}}
\newcommand{\tSAND}{\mathtt{SAND}}
\newcommand{\BAS}{\mathtt{BAS}}
\newcommand{\R}[1]{\recht{R}_{#1}}
\newcommand{\AT}{\recht{AT}}
\newcommand{\OAT}{\underline{\AT}}
\newcommand{\BB}{\mathbb{B}}

\newcommand{\MBool}{\recht{MBool}}
\newcommand{\OMBool}{\underline{\MBool}}

\newcommand{\ch}{\recht{ch}}
\newcommand{\BU}{\mathtt{BU}}

\newcommand{\cl}[1]{[#1]}
\newcommand{\tzero}{\mathtt{0}}
\newcommand{\tone}{\mathtt{1}}
\newcommand{\struc}[1]{\recht{S}_{#1}}
\newcommand{\phstr}{\varphi^{\recht{struc}}}
\newcommand{\pp}{\mathtt{p}}
\newcommand{\oo}{\mathtt{o}}
\newcommand{\tC}{\mathtt{C}}
\newcommand{\ODAT}{\underline{\recht{DAT}}}
\newcommand{\Bic}{\recht{Bic}}
\newcommand{\OBic}{\underline{\Bic}}
\newcommand{\OADT}{\underline{\recht{ADT}}}

\newtheorem{fact}{Fact}[section]
\newtheorem{definition}[fact]{Definition}
\newtheorem{theorem}[fact]{Theorem}
\newtheorem{lemma}[fact]{Lemma}
\newtheorem*{question}{Question}

\theoremstyle{definition}
\newtheorem{example}[fact]{Example}
\newtheorem{remark}[fact]{Remark}

\newcommand{\hlbox}[1]{%
  \smallskip\begin{center}
  \fboxrule1pt\fboxsep3pt\fcolorbox{black!45}{black!8}{%
  \begin{minipage}{.96\linewidth}#1\end{minipage}}
  \end{center}\smallskip}

  \newcommand*{\ldb}{\{\mskip-5mu\{}
\newcommand*{\rdb}{\}\mskip-5mu\}}

\allowdisplaybreaks

\begin{document}
\title{Attack tree metrics are operad algebras}

\author{Milan Lopuhaä-Zwakenberg\\
\textit{University of Twente}\\
Enschede, the Netherlands \\
m.a.lopuhaa@utwente.nl
}

\maketitle

\begin{abstract}
Attack Trees (ATs) are a widely used tool for security analysis. ATs can be employed in quantitative security analysis through metrics, which assign a security value to an AT. Many different AT metrics exist, and there exist multiple general definitions that aim to study a wide variety of AT metrics at once. However, these all have drawbacks: they do not capture all metrics, and they do not easily generalize to extensions of ATs. In this paper, we introduce a definition of AT metrics based on category theory, specifically operad algebras. This encompasses all previous definitions of AT metrics, and is easily generalized to extensions of ATs. Furthermore, we show that under easily expressed operad-theoretic conditions, existing metric calculation algorithms can be extended in considerable generality.	
\end{abstract}

\section{Introduction}

\begin{wrapfigure}[18]{r}{4cm}
\vspace{-1em}
\centering
\begin{tikzpicture}[
and/.style={and gate US,rotate=90,draw,fill = lightgray},
or/.style={or gate US,rotate=90,draw,fill = lightgray},
bas/.style={circle,draw,fill = lightgray}]
\draw (-1.25,2) node[or]{};
\draw (0.25,2) node[and]{} ;
\draw[fill = lightgray] (1.75,2) circle (0.3cm);
\draw (-1.25,1.4) node{$\tOR$};
\draw (0.25,1.4) node{$\tAND$};
\draw (1.75,1.4) node{$\BAS$};
\draw (0,0) node[or] (r) {\rotatebox{270}{$r$}};
\draw (0,0.45) node[rectangle,draw,fill=white] {\tiny rob bank};
\draw (-0.5,-1) node[bas] (f) {$f$};
\draw (0.5,-1) node[and] (s) {\rotatebox{270}{$s$}};
\draw (0,-2) node[bas] (b) {$b$};
\draw (1,-2) node[bas] (l) {$l$};
\draw (f) -- (r) -- (s) -- (b);
\draw (s) -- (l);
\draw (-0.5,-0.55) node[rectangle,draw,fill=white] {\tiny by force};
\draw (0.5,-0.55) node[rectangle,draw,fill=white] {\tiny steal};
\draw (-0.2,-1.6) node[rectangle,draw,fill=white] {\tiny break in};
\draw (1.2,-1.6) node[rectangle,draw,fill=white] {\tiny lockpicks};
\end{tikzpicture}
\caption{Attack tree of an attacker robbing a bank ($r$). They can either take the money by force ($f$), or they steal the money ($s$) by purchasing lockpicks for the vault ($l$) and breaking in at night ($b$).}
\label{fig:bank1}
\end{wrapfigure}
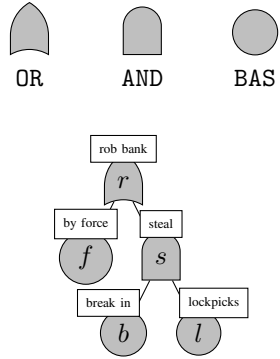

Attack trees (ATs) are hierarchical diagrams that categorize and classify the vulnerabilities of a (high-tech) system. ATs have been employed in a variety of settings, such as railway control systems \cite{dong2017attack}, smart grids \cite{beckers2014determining} and nuclear control systems \cite{khand2007attack}. ATs consist of \emph{basic attack steps} (BASs), indivisible basic actions which can be performed by an adversary, as well as $\tAND$/$\tOR$-gates whose activation depends on the activation of their children. An attack is considered succesful if the top node is activated. Note that an AT is not necessarily a tree, as a node may have multiple parents. Besides `standard' ATs numerous extensions exists, usually with extra gates, such as sequential AND-gates \cite{jhawar2015attack} or defenses \cite{kordy2010foundations}.

Apart from the qualitative analysis of categorizing all attacks, ATs can also be used for quantitative analysis, in which a system's vulnerability is expressed by a security value. This is usually done by assigning a value to each BAS, and then using these to calculate the system's value. For instance, each BAS can be assigned a cost value representing the resources the attacker has to spend to perform this BAS; given these, the \emph{min cost} metric is then the minimal cost of a succesful attack \cite{schneier1999attack}. There are many other such metrics, such as the minimal time of a succesful attack, mean time to compromise, or attack damage \cite{mcqueen2006time}. Because many AT metrics are relevant for security analysis, there is a demand for a general framework in which AT metrics can be formulated and computed.

\begin{question}
What is a general definition for AT metrics, and what algorithms exist to calculate general AT metrics?
\end{question}

\begin{wrapfigure}[18]{r}{4cm}
\label{fig:mintime}
\centering
\vspace{-1em}
\begin{tikzpicture}[
and/.style={and gate US,rotate=90,draw,fill = lightgray},
or/.style={or gate US,rotate=90,draw,fill = lightgray},
bas/.style={circle,draw,fill = lightgray}]
\draw (0,0) node[and] (r) {\rotatebox{270}{\small $\rightarrow$}};
\draw (0,-1) node [bas] (a) {{\small $a$}};
\draw (a) edge[bend right]  (r);
\draw (a) edge[bend left]  (r);
\draw (1.5,-1) -- (2.5,0);
\draw (1.5,-1) node [bas] (b) {\small $b$};
\draw (2.5,-1) node [bas] (c) {\small $c$};
\draw (1.5,0) node [and] (a1) {\rotatebox{270}{\small $\rightarrow$}};
\draw (2.5,0) node [and] (a2) {};
\draw (2,1) node [or] (o) {};
\draw (b) -- (a1) -- (o) -- (a2) -- (c);
\draw (c) -- (a1);
\draw (0,-1.5) node {$T_1$};
\draw (2,-1.5) node {$T_2$};
\end{tikzpicture}

\begin{tabular}{c|cc}
& $T_1$ & $T_2$ \\ \hline
\cite{jhawar2015attack} & $2t_a$ & $\max\{t_b,t_c\}$ \\
\cite{kumar2015quantitative} & $t_a$ & $\max\{t_b,t_c\}$ \\
\cite{9925106} & $\infty$ & $t_b+t_c$ \\
\cite{lopuhaa2021attack} & $\infty$ & $\max\{t_b,t_c\}$
\end{tabular}

\caption{The conflicting definitions of the \emph{min time} metric in previous work, for dynamic ATs $T_1$ and $T_2$. The time of BAS $a$ is denoted $t_a$.}
\end{wrapfigure}
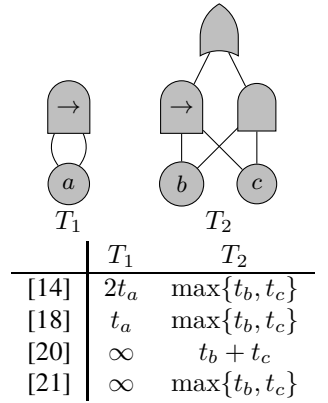

Approaches to a generic formalization of AT metrics have been proposed in the literature \cite{mauw2005foundations,bossuat2017evil,9925106}. These assume the metric to take values in a \emph{semiring}: a set in which the BASs take their value (e.g., $\mathbb{R}_{\geq 0}$ for \emph{min cost}), that has two operations satisfying certain axioms \cite{mauw2005foundations}, which are used to define the AT's security value in terms of that of the BASs. However, different works have different ways of defining the AT metric in terms of the BAS values, leading to different, noncompatible definitions of the same metric. For instance, the min time metric of dynamic ATs (with a sequential AND gate) has four different definitions in the literature, that are incompatible even for small examples, see Fig.~2. In particular, there is not one definition of \emph{semiring metrics}, but at least three noncompatible ones, see Table 1.

The fact that many metrics have different, noncompatible definitions is likely due to the fact many papers define their metric along with a computation algorithm: \emph{min time} as defined in \cite{jhawar2015attack} comes with a bottom-up algorithm, while the \emph{min time} of \cite{kumar2015quantitative} is calculated via priced-timed automata. As a result, metrics are often defined in a way to fit the algorithm, at a cost of possibly being inconsistent with earlier literature. Another approach is to define the metric directly from the semantics \cite{budde2021efficient,lopuhaa2021attack}; calculating these metrics is NP-hard.

Besides the existence of noncompatible metric formalizations, there are two other issues with existing AT metric formalizations. First, there are still metrics that do not fit into any existing framework, such as total attack probability \cite{rauzy1997exact}. Second, semiring formalizations are only defined for standard ATs (just $\tAND/\tOR$-gates). Extending them to extensions of ATs typically requires a completely new definition \cite{kordy2018quantitative,9925106}. 

Summarizing the previous paragraphs, we conclude that there is a need for a formal framework for AT metrics that is generic enough to capture all existing frameworks and metrics, and has straightforward extensions to extensions of ATs such as dynamic ATs and attack-defense trees.

% The semiring metric definitions have some downsides. First, they are too restrictive: there are several metrics, such as total attack probability \cite{rauzy1997exact}, that do not fit either definition. Second, the definitions are incompatible: the same metric calculates differently under one definition compared another, and may be different again compared to works considering only one metric.  Third, semiring approaches are defined for standard ATs (just $\tAND$/$\tOR$-gates), and extending them to extensions of ATs typically requires a completely new definition \cite{kordy2018quantitative,9925106}. A framework in which extensions can easily be incorporated would be preferred.

\noindent \textbf{Contributions.} In this paper, we create a single, comprehensive framework for AT metrics based on category theory. This framework aims to encompass all sensible metrics; while it is possible to construct pathological counterexamples, we show that the necessary and sufficient conditions for a metric to fall within our framework are completely natural, and indeed are met by almost all metrics from the literature. Hence, rather than adding yet another row to the table in Fig.~2, we provide an overarching definition of which all existing definitions are special cases. Furthermore, this new framework has straightforward extensions to extensions of the AT formalism. An overview of existing metrics and metric formalisms that are now unified under a single  framework is given in Table I.

Existing definitions of AT metrics define how, given a value $x_a$ in a set $X$ for each BAS $a$, and an AT $T$, one calculates a metric value $x_T \in X$. We shift perspective, and in this paper a metric assigns to each AT $T$ with $n$ BASs a function $\varphi_T\colon X^n \rightarrow X$ that maps the BAS values $\vec{x}$ to $x_T$. This is more easily understood in terms of \emph{operads} \cite{markl2002operads}, which generalize sets of multivariate functions on a set. The sets $\End_n(X)$ of maps $X^n \rightarrow X$, together with a natural composition operation, form an operad $\OEnd(X)$. We show that ATs form an operad $\OAT$; thus we define an AT metric as an \emph{operad morphism} $\varphi\colon \OAT \rightarrow \OEnd(X)$, i.e., $X$ is an \emph{operad algebra}. This approach is closely related to other operad models of treelike structures, such as phylogenetic trees \cite{baez2015operads} and proof trees \cite{meseguer1989general}. Operads are a natural tool to formalize treelike structures, since their composition operation mirrors the hierarchical nature of trees. Note that $X$ need not be a set but can be an object in a general category (e.g. topological spaces), which implies special properties of the function $\varphi_T$ (e.g. continuity).

\begin{table*}
\centering
\begin{tabular}{ll|ccccccc|c}
metric type &  source & cost & time & damage & probability & skill & satisfiability & Pareto fronts & Operad metric? \\ \hline 
\multicolumn{10}{c}{standard ATs} \\ \hline
Total attack probability & \cite{rauzy1997exact} &&&&+&&&& yes \\
propositional semiring metrics & \cite{9925106} & +&+&+&+&+&+&+& yes \\
bottom-up semiring metrics & \cite{mauw2005foundations}&+&+&+&+&+&+&& yes \\
set semantics semiring metrics & \cite{bossuat2017evil}&+&+&+&+&+&+&+& yes \\
mean time to compromise & \cite{mcqueen2006time} &&+&&&&&&yes/no
\\ \hline
\multicolumn{10}{c}{Dynamic ATs} \\ \hline
well-formed propositional dynamic semiring & \cite{budde2021efficient}& & +&&&&&& yes \\
propositional dynamic semiring & \cite{lopuhaa2021attack,9925106} && +&&&&&& yes   \\
bottom-up dynamic semiring &  \cite{jhawar2015attack} && +&&&&&& yes \\
priced-timed-automata & \cite{kumar2015quantitative}& + & + & + &&&&+& yes\\ \hline
\multicolumn{10}{c}{Attack-defense trees} \\ \hline
bottom-up attribute domain &\cite{kordy2018quantitative,fila2019efficient} & + & + & + & + & + & +&+& yes \\
set semantics attribute domain & \cite{kordy2018quantitative} & + & + & + & + & + & +&+& no \\
\end{tabular}
\caption{Attack tree metric formalizations from the literature and the metrics that fall under their framework. Different entries in one column are noncompatible definitions. Most of these are unified under the operad algebra framework; for \emph{mean time to compromise} this comes at a loss of computational efficiency, see Section \ref{sec:ubi}.}
\end{table*}

Beyond the definition the operad algebra framework, we also give a list of necessary and sufficient conditions for an AT metric to be considered an operad algebra. As our metric is very general, these are almost always satisfied. Showing that a metric is an operad algebra has two advantages, the first being that it acts as a type of sanity check: the only nontrivial axiom for an operad algebra is that the metric should respect the hierarchical structure of attack trees. This is desired behaviour for many metric definitions, and operad algebras form a natural formalization for this.

Second, expressing a metric in the operad algebrad framework gives access to the metric calculation algorithms developed in this paper. By design, our definition is very general, and we cannot hope for a  fast general algorithm for all metrics, especially since this is known to be NP-complete \cite{9925106}. Instead, we formulate two existing algorithms, the bottom-up algorithm \cite{mauw2005foundations} and a BDD-based algorithm \cite{9925106}, in the language of operad algebras, and we give sufficient conditions on metrics as to when they can be applied. Both are state-of-the-art for the metrics to which they apply; thus there is no loss in efficiency incurred by the increased generality of our framework.

As a final contribution, this paper also showcases how the operad algebra framework can be extended to extensions of the AT formalism, by showing how this is done for dynamic ATs \cite{jhawar2015attack} and attack-defense trees \cite{kordy2010foundations}. These extensions also generalize a number of metrics as summarized in Table I.

% This definition has several advantages: First, operad morphisms put only a few assumptions on the maps $\varphi_T$, which makes this framework widely applicable. The only axiom states that $\varphi_T$ should allow for \emph{modular analysis} \cite{dutuit1996linear}: When $T$ is obtained by replacing a BAS in a smaller AT $T'$ by another AT $T''$, the function $\varphi_T$ can be expressed in terms of $\varphi_{T'}$ and $\varphi_{T''}$. This holds for all metrics above, and it also fits in the AT philosophy, where a BAS can be refined into a sub-AT if more detailed analysis is needed. Second, the operad $\OAT$ can easily be extended to incorporate more gate types; thus we immediately gain definitions for metrics for AT extensions. Third, instead of a set, $X$ can be an object in a general category. Many properties (continuity, linearity,...) of the functions $\varphi_T$ can be stated by saying that the metric takes values in a certain category (topological spaces, real vector spaces,...).

% Another important aspect of AT metrics is their computation. By design, our definition is very general, and we cannot hope for a  fast general algorithm for all metrics, especially since this is known to be NP-complete \cite{9925106}. Instead, we formulate two existing algorithms, the bottom-up algorithm \cite{mauw2005foundations} and a BDD-based algorithm \cite{9925106}, in general terms compaitlbe with our AT metric definition, and we give sufficient operad-theoretic conditions as to when these algorithms work. For the bottom-up method, these conditions are actually necessary.

\hlbox{%
\emph{Contributions:}
We create generalized operad-theoretical framework for attack tree metrics which:
\begin{enumerate}
\item encompasses a wide variety of existing AT metrics;
\item can straightforwardly be extended to extensions of the AT formalism;
\item can be used in a general, category-theoretical setting;
\item has easy to understand necessary and sufficient conditions;
\item allows for a generalization of state-of-the-art metric computation algorithms to a wide variety of metrics.
\end{enumerate}
}

\subsection{Related work}

There are three approaches to defining AT metrics in terms of semirings. The first one \cite{mauw2005foundations} defines a metric bottom-up: At every node, a value is calculated based on its childrens' values, with the operation determined by the type of gate. The AT's metric value is the top node's value. This approach has also been extended to dynamic ATs \cite{jhawar2015attack} and attack-defense trees \cite{kordy2010foundations}. By definition, these metrics are quickly  computed via a bottom-up algorithm. However, this model assumes that a node must be activated once for each parent. It can be argued that this limits the scenarios that can be modeled by ATs \cite{lopuhaa2021attack}.

Another approach is to first define a metric value to each (minimal) succesful attack, and then combine them into a single metric value for the AT \cite{9925106}. For instance, for \emph{min cost} one first calculates the cost of each succesful attack, and then takes the minimum over all of these. This too can be extended to attack-defense trees \cite{kordy2018quantitative}, and with a bit more effort to dynamic ATs \cite{9925106}. This approach is more flexible than the bottom-up approach, but calculating such metrics is NP-hard in general. In \cite{9925106} a BDD-based approach to metric calculation is proposed, which is still worst-case exponential, but considerably faster in practice. The BDD approach has also been applied to total attack probability \cite{bryant1992symbolic}. The third approach to semiring metrics \cite{bossuat2017evil} is a middle ground between the first two: first a candidate set of attacks is created bottom-up, and then this set is used as the set of minimal attacks in the previous approach. The created set contains the set of minimal attacks, so these metrics coincide with those of \cite{9925106} when minimal metrics are known to be optimal, which is true for most but not all semirings.

Operads and their algebras originate in algebraic topology, but have many applications in other fields \cite{markl2002operads}. In computer science, they have been applied to wiring diagrams \cite{yau2018operads}, design specification \cite{foley2021operads}, and information theory \cite{bradley2021entropy}. Operads have also been used to model treelike structures such as phylogenetic trees \cite{baez2015operads} and proof trees \cite{meseguer1989general}.
Finally, ATs are closely related to fault trees, which are used in safety \cite{limnios2013fault}. Since the mathematical formalisms are the same, the results of this paper can also be applied to fault tree metrics such as \emph{mean time to fail}.

\noindent \textbf{Notation.} Isomorphisms are written $\iso$. 
%Disjoint products are $\sqcup,\coprod$.
Categories are $\cat{Set},\cat{Pos},\cat{Top}$,... The set $\{1,\ldots,n\}$ is written $\cl{n}$ for $n \in \mathbb{Z}_{\geq 0}$. We let $\BB = \{\tzero,\tone\}$. Multisets are denoted $\ldb \ldots \rdb$, with multiset union denoted $\uplus$.

\section{Attack trees}

In this section, we review the formalism of ATs. They are defined as follows:

\begin{definition} \label{def:at}
An \emph{attack tree} is a triple $T = (N,E,\gamma)$, where $(N,E)$ is a rooted directed acyclic multigraph and $\gamma$ is a function $\gamma\colon N \rightarrow \{\tOR,\tAND,\BAS\}$, such that $\gamma(v) = \BAS$ if and only if $v$ is a leaf of $(N,E)$.
\end{definition}

\emph{Multigraph} here means that $E$ is a multiset in $N \times N$, meaning there can be multiple edges between two nodes. We need this to model ATs such as $T_1$ in Fig.~2, as well as the constructs that come up when we merge BASs in Section \ref{sec:bu}.

The root of $T$ is denoted $\R{T}$, and the children of a node $v$ form the multiset $\ch(v) = \ldb w \in N \mid (v,w) \in E \rdb$. If $\ch(v) = \ldb v_1,\ldots,v_n \rdb$, we also write $v = \tOR(v_1,\ldots,v_n)$ or $\tAND(v_1,\ldots,v_n)$ depending on $\gamma(v)$; the order of the $v_i$ does not matter here. The set of BASs is denoted $B_T$; we write $B$ if there is no confusion. For $v \in N$, we let $T_v$ be the full subDAG of $T$ with root $v$, i.e., the DAG consisting of all descendants of $v$. Note that despite their names, ATs are not necessarily trees.

%One may choose to allow multiple edges between two nodes, making $E$ a multiset and $(N,E)$ a directed multigraph. This is useful to model ATs such as $T_1$ from Fig.~2 in which one node has the same child multiple times. To simplify notation, in what follows we disallow multiple edges and consider $E$ to be a set, in order to simplify notation; however, all results are readily generalized to directed multigraphs.

The role of an AT is to identify possible attacks on a system. The set of attacks on $T$ with $B_T = \{a_1,\ldots,a_n\}$ is $\BB^n$, where $\vec{b} \in \BB^n$ represents the attack caused by activating all $a_i$ for which $b_i = \tone$. The effect of an attack on the AT is given by its \emph{structure function} $\struc{T}$, where $\struc{T}(\vec{b},v) \in \mathbb{B}$ expresses whether the attack $\vec{b}$ activates the node $v$:

\begin{definition} \label{def:struc}
Let $T$ be an AT with $B_T = \{a_1,\ldots,a_n\}$. The \emph{structure function} of $T$ is the function $\struc{T}\colon \BB^n \times N \rightarrow \BB$ defined as
\begin{equation*}
\struc{T}(\vec{b},v) = \begin{cases} 
b_i, & \textrm{ if $v = a_i$},\\
\bigwedge_{w \in \ch(v)} \struc{T}(\vec{b},w), & \textrm{ if $\gamma(v) = \tAND$},\\
\bigvee_{w \in \ch(v)} \struc{T}(\vec{b},w), & \textrm{ if $\gamma(v) = \tOR$}.
\end{cases}
\end{equation*}
\end{definition}

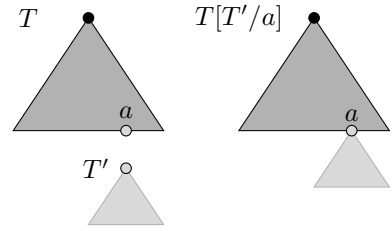
\begin{figure}
\centering
\begin{tikzpicture}
\draw[fill = black!30] (6,0) -- (5,-1.5) -- (7,-1.5) -- cycle;
\draw[black!30, fill = black!15] (6.5,-1.5) -- (7,-2.25) -- (6,-2.25) -- cycle;
\filldraw (6,0) circle (2pt);
\draw[fill = black!15] (6.5,-1.5) circle (2pt);
\draw (5,0) node {$T[T'/a]$};
%\draw (0.3,0) node {$\recht{R}_T$};
\draw (6.5,-1.3) node {$a$};
\draw[fill = black!30] (3,0) -- (2,-1.5) -- (4,-1.5) -- cycle;
\filldraw (3,0) circle (2pt);
\draw[fill = black!15] (3.5,-1.5) circle (2pt);
\draw (2.2,0) node {$T$};
%\draw (3.3,0) node {$\recht{R}_T$};
\draw (3.5,-1.25) node {$a$};
\draw[black!30, fill = black!15] (3.5,-2) -- (4,-2.75) -- (3,-2.75) -- cycle;
\draw[fill = black!15] (3.5,-2) circle (2pt);
%\draw (3.8,-2) node {$v$};
\draw (3.1,-2) node {$T'$};
\end{tikzpicture}
\caption{Modular composition of ATs.} \label{fig:mod}
\end{figure}

\begin{definition} \label{def:semantics}
An attack $\vec{b} \in \BB^n$ is called \emph{successful} if $\struc{T}(\vec{b},\R{T}) = \tone$. The set of successful attacks on $T$ is denoted $\recht{Suc}_T$.
\end{definition}

An important way of combining ATs is \emph{modular composition} \cite{dutuit1996linear}: in an AT $T$, a BAS $a$ is replaced by an entire AT $T'$, to obtain the larger AT $T[T'/a]$, see Fig.~\ref{fig:mod}. This is used, for instance, when in an AT created by a domain expert a BAS $a$, which was first considered its own indivisible event, is given its own AT to describe its failure conditions. Formally:

\begin{definition} \label{def:mod}
Let $T = (N,E,\gamma)$ and $T' = (N',E',\gamma')$ be ATs with $N \cap N' = \varnothing$, and let $a \in N$ such that $\gamma(a) = \BAS$. Then the \emph{modular composition $T[T'/a] = (N'',E'',\gamma'')$} is the AT defined as
\begin{align*}
N'' &= (N \setminus \{a\}) \cup N',\\
\gamma''(w) &= \begin{cases}
\gamma(w), & \textrm{ if $w \in N$},\\
\gamma'(w), & \textrm{ if $w \in N'$},
\end{cases}\\
E'' &= \ldb (v,w) \in E \mid w \neq a \rdb \uplus E'\\
& \quad \quad \quad \uplus \ldb (w,\R{T'}) \mid w \in N, (w,a) \in E \rdb.
\end{align*}
\end{definition}

Besides categorizing succesful attacks on a system, ATs can also be used in quantitative security analysis to calculate security metrics, such as the minimal time or cost an attacker needs for a succesful attacks. Such metrics work by assigning a certain security value $\alpha_a$ to each BAS, and using these to calculate a security value $\alpha_T$ for the AT. However, there are many ways of defining $\alpha_T$ in terms of the $\alpha_a$, leading to noncompatible definitions of AT metrics in the literature. In the following sections, we introduce an overarching framework.

The concept of ATs has been extended by multiple works to meet the demands of more elaborate forms of security analysis. These extensions typically define new types of gates, such as defense gates \cite{kordy2010foundations} or sequential $\tAND$-gates \cite{jhawar2015attack,9925106}. Such extensions often require completely new definitions of metrics.

\section{Operads}

In this section we will define operads and their morphisms. Operads arise as a way to axiomatize collections of maps $X^n \rightarrow X$ for a set $X$ \cite{markl2002operads}. For $n \in \mathbb{Z}_{\geq 0}$ and a set $X$, consider the sets of maps $\End_n(X) := X^n \rightarrow X$ (so $\End_0(X) = X$). As we range over $n$, the sets $\End_n(X)$ have the following structures:
\begin{enumerate}
\item For each permutation $\sigma\colon [n] \iso [n]$ there is a map $\tau_{\sigma}\colon \End_n(X) \iso \End_{n}(X)$ defined as follows: for $f \colon X^n \rightarrow X$ and $\vec{x} \in X^n$ one has $\tau_{\sigma}(f)(\vec{x}) = f(x_{\sigma(1)},\ldots,x_{\sigma(n)})$, see Fig.~\ref{fig:tau}.
\item There is an element $\id \in \End_1(X)$ representing the identity $X \iso X$.
\item Let $n \in \mathbb{Z}_{\geq 1}$ and let $m_1,\ldots,m_n \in \mathbb{Z}_{\geq 0}$; and let $f \in \End_n(X)$ and $g_i \in \End_{m_i}(X)$ for all $i$. We then define the composition $f \star \vec{g}\in \End_{\sum_i m_i}(X)$ by first applying each $g_i$ to its own $m_i$ arguments, and then applying $f$ to the $n$ resulting outcomes. More formally, for all $x_{1,1},\ldots,x_{1,m_1},\ldots,x_{n,1},\ldots,x_{n,m_n} \in X$:
\begin{align*}
&(f \star \vec{g})(x_{1,1},\ldots,x_{n,m_n})\\
&= f(g_1(x_{1,1},\ldots,x_{1,m_1}),\ldots,g_n(x_{n,1}\ldots,x_{n,m_n})).
\end{align*}
\end{enumerate}

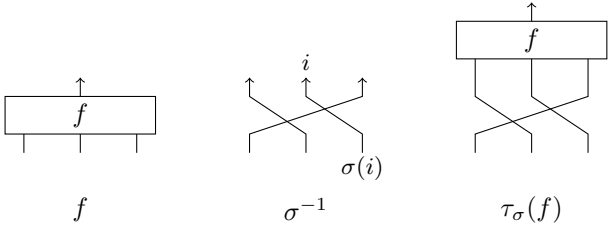
\begin{figure}
\centering
\begin{tikzpicture}
\draw (0,0) -- (2,0) -- (2,-0.5) -- (0,-0.5) -- cycle;
\node at (1,-0.25) {$f$};
\draw[->] (1,0) -- (1,0.25);
\draw (0.25,-0.5) -- (0.25,-0.75);
\draw (1,-0.5) -- (1,-0.75);
\draw (1.75,-0.5) --  (1.75,-0.75);
\node at (1,-1.5) {$f$};

\draw[<-] (3.25,0.25) -- (3.25,0) -- (4,-0.5) -- (4,-0.75);
\draw[<-] (4,0.25) -- (4,0) -- (4.75,-0.5) -- (4.75,-0.75);
\draw[<-] (4.75,0.25) -- (4.75,0) -- (3.25,-0.5) -- (3.25,-0.75);
\node at (4, 0.45) {\small $i$};
\node at (4.75, -0.95) {\small {$\sigma(i)$}};
\node at (4,-1.5) {$\sigma^{-1}$};

\draw (6.25,0.5) -- (6.25,0) -- (7,-0.5) -- (7,-0.75);
\draw (7,0.5) -- (7,0) -- (7.75,-0.5) -- (7.75,-0.75);
\draw (7.75,0.5) -- (7.75,0) -- (6.25,-0.5) -- (6.25,-0.75);
\node at (7,-1.5) {$\tau_{\sigma}(f)$};
\draw (6,1) -- (8,1) -- (8,0.5) -- (6,0.5) -- cycle;
\draw[->] (7,1) -- (7,1.25);
\node at (7,0.75) {$f$};
\end{tikzpicture}
\caption{The map $\tau_{\sigma}(f)$ represented diagrammatically: first $\sigma^{-1}$ is applied to permute the arguments, and then $f$ is applied to the permuted arguments.} \label{fig:tau}
\end{figure}

The constructs $\tau_{\sigma}$, $\id$, $\star$ satisfy a number of different properties; for instance $\id \star f = f$ for each $f \in \End_n(X)$. The idea behind operads is to take the collection of all $\End_n(X),\tau_\sigma,\id,\star$ as the primordial object, rather than the set $X$; we call this new object $\OEnd(X)$, the \emph{operad} of maps on $X$. Thus, an operad can be thought of as a collection of maps, but examples will show that the definition is more general.

\begin{definition}\label{def:operad}
An \emph{operad} is a tuple $\underline{O} = (O,\tau,\id,\star)$ consisting of: a set $O_n$ for each $ n \in \mathbb{Z}_{\geq 0}$; a bijection $\tau_{\sigma}\colon O_n \iso O_n$ for each bijection $\sigma\colon [n] \iso [n]$; an element $\id \in O_1$; and for each $n,m_1,\ldots,m_n \in \mathbb{Z}_{\geq 0}$, a map $\star\colon O_n \times \prod_{i=1}^n O_{m_i} \rightarrow O_{\sum_i m_i}$. These must satisfy the following axioms:
\begin{enumerate}
    \item For each $\sigma,\sigma'\colon [n] \iso [n]$ one has $\tau_{\sigma \circ \sigma'} = \tau_{\sigma} \circ \tau_{\sigma'}$;
    \item For each $n$ and $f \in O_n$ one has $\id \star f = f \star (\id,\ldots,\id) = f$;
    \item For all $n,m_1,\ldots,m_n,k_{1,1},\ldots,k_{n,m_n} \in \mathbb{Z}_{\geq 0}$ and for all $f \in O_n$, $g_{i} \in O_{m_i}$ and $h_{i,j} \in O_{k_{i,j}}$ one has ($i \leq n$ and $j \leq m_i$)
    \begin{align*}
    &f \star (g_1 \star (h_{1,1},\ldots,h_{1,m_1}),\ldots,g_n \star (h_{n,1},\ldots,h_{n,m_n}))\\
    &= (f \star \vec{g}) \star \vec{h};
    \end{align*}
    \item For all $n,m_1,\ldots,m_n \in \mathbb{Z}_{\geq 0}$, $f \in O_n$, $g_i \in O_{m_i}$, and $\sigma_i\colon [m_i] \iso [m_i]$ one has
    \begin{equation*}
    f \star (\tau_{\sigma_1}(g_1),\ldots,\tau_{\sigma_n}(g_n)) = \tau_{(\sigma_1,\ldots,\sigma_n)}(f \star \vec{g}).
    \end{equation*}
    where $(\sigma_1,\ldots,\sigma_n)\colon \left[\sum_i m_i\right] \rightarrow \left[\sum_i m_i\right]$ is given by $(\sigma_1,\ldots,\sigma_n)\left(\sum_{i'< i} m_{i'}+k\right) = \sum_{i' < i} m_{i'} + \sigma_{i}(k)$, for all $i \leq n$ and $k \leq m_{i+1}$.
    \item For all $n,m_1,\ldots,m_n \in \mathbb{Z}_{\geq 0}$, $f \in O_n$, $g_i \in O_{m_i}$, and $\sigma\colon [n] \iso [n]$ one has
    \begin{equation*}
    \tau_{\sigma}(f) \star \vec{g} = f \star (g_{\sigma(1)},\ldots,g_{\sigma(n)});
    \end{equation*}
\end{enumerate}
\end{definition}

Most of these axioms are understandable from the perspective of $\End_n(X)$. In 4), the permutation $(\sigma_1,\ldots,\sigma_n)$ keeps the $n$ `blocks' of $m_i$ consecutive elements in place, while permuting each block via $\sigma_i$. The equality is then shown for $\OEnd(X)$ by the diagram of Fig.~\ref{fig:tau2}.

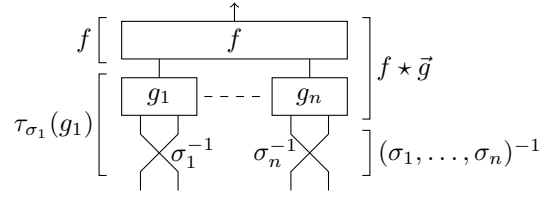
\begin{figure}
\centering
\begin{tikzpicture}
\draw (0,0) -- (3,0) -- (3,-0.5) -- (0,-0.5) -- cycle;
\node at (1.5,-0.25) {$f$};
\draw[->] (1.5,0) -- (1.5,0.25);
\draw (0.5,-0.5) -- (0.5,-.75);
\draw (0,-.75) -- (1,-.75) -- (1,-1.25) -- (0,-1.25) -- cycle;
\node at (0.5,-1) {$g_1$};
\draw (0.25,-1.25) -- (0.25,-1.5) -- (0.75,-2) -- (0.75,-2.25);
\draw (0.75,-1.25) -- (0.75,-1.5) -- (0.25,-2) -- (0.25,-2.25);
\node at (0.95,-1.75) {$\tiny \sigma_1^{-1}$};
\draw[dashed] (1.1,-1) -- (1.9,-1);
\draw (2.5,-0.5) -- (2.5,-.75);
\draw (2,-.75) -- (3,-.75) -- (3,-1.25) -- (2,-1.25) -- cycle;
\node at (2.5,-1) {$g_n$};
\draw (2.25,-1.25) -- (2.25,-1.5) -- (2.75,-2) -- (2.75,-2.25);
\draw (2.75,-1.25) -- (2.75,-1.5) -- (2.25,-2) -- (2.25,-2.25);
\node at (2.05,-1.75) {$\tiny \sigma_n^{-1}$};
\draw (-0.2,0.05) -- (-0.3,0.05) -- (-0.3,-0.55) -- (-0.2,-0.55);
\draw (-0.2,-0.7) -- (-0.3,-0.7) -- (-0.3,-2.05) -- (-0.2,-2.05);
\draw (3.2,0.05) -- (3.3,0.05) -- (3.3,-1.3) -- (3.2,-1.3);
\draw (3.2,-1.45) -- (3.3,-1.45) -- (3.3,-2.05) -- (3.2,-2.05);
\node at (-0.5,-0.25) {$f$};
\node at (-0.9,-1.375) {$\tau_{\sigma_1}(g_1)$};
\node at (3.75,-.625) {$f \star \vec{g}$};
\node at (4.5,-1.75) {$(\sigma_1,\ldots,\sigma_n)^{-1}$};
\end{tikzpicture}
\caption{Axiom 4) of Definition \ref{def:operad} for the endomorphism operad $\OEnd(X)$: Applying $f$ to arguments of the form $\tau_{\sigma_i}(g_i)$ yields the same function as first permuting all arguments via $(\sigma_1,\ldots,\sigma_{n})^{-1}$ and then applying $f \star \vec{g}$.} \label{fig:tau2}
\end{figure}

\begin{example} For $n \in \mathbb{Z}_{\geq 0}$, define $\Delta_n \subset \mathbb{R}^n_{\geq 0}$ by
\[\Delta_n = \left\{\vec{p} \in \mathbb{R}_{\geq 0}^n \ \middle| \ \sum_i p_i = 1\right\}\]
(so $\Delta_{0} = \varnothing$); that is, $\Delta_n$ is the set of probability distributions on $[n]$. Furthermore, define 
\begin{equation*}
\vec{p} \star (\vec{q}_1,\ldots,\vec{q}_n) = (p_1q_{1,1},\ldots,p_1q_{1,m_1},\ldots,p_nq_{n,m_n}).
\end{equation*}
One can interpret $p \star \vec{q}$ as a two-step probabilistic process: first one draws $X \in [n]$ via probability distribution $p$; and then one draws $Y \in [m_X]$ via probability distribution $q_X$. The resulting random variable $Y \in \left[\sum_i m_i\right]$ has probability distribution $p \star \vec{q}$. Upon defining $\id$ and $\tau_{\sigma}$ in a straightforward way, one obtains the \emph{probability operad} $\underline{\Delta}$.
\end{example}

\begin{example}
The construction of the operad $\OEnd(X)$ from the beginning of this section can also be done in any symmetric monoidal category $(\CC,\otimes)$. For an object $X \in \ob(\CC)$, define $\End_n(X) := \Hom_{\CC}(X^{\otimes n},X)$. Then $\id$ is again the identity map, and $\tau_{\sigma}$ is precomposition with the induced isomorphism $X^{\otimes n} \iso X^{\otimes n}$ that sends the $\sigma(i)$-th component to the $i$-th component. Furthermore, $f \star \vec{g}$ is defined to be the composition
\[
X^{\otimes \sum_i m_i} \iso \bigotimes_{i=1}^n X^{\otimes m_i} \xrightarrow{\bigotimes_i g_i} X^{\otimes n} \stackrel{f}{\longrightarrow} X.
\]
This defines the operad $\OEnd_{\CC}(X)$. We write $\OEnd(X)$ when $\CC$ is clear. In this paper, $\otimes$ will always be the binary product (in a category that has these).
\end{example}

\begin{remark}
Some authors define a set $O_I$ for each finite set $I$, representing the maps $X^I \rightarrow X$ \cite{doubek2017properads}. This is `cleaner' in the sense that it does not assume an ordering on the arguments; the cost is paid in having a slightly more elaborate definition of $\tau_{\sigma}$. The two definitions are equivalent via the identification $O_n = O_{[n]}$. We choose our definition as it is more convenient for our metric computation algorithms later on, as well as to emphasize that we think of AT metrics as $n$-ary functions.
\end{remark}

\subsection{Operad algebras}

The motivation behind defining operads was to axiomatize $\OEnd(X)$. The link between operads and sets of functions can be made more direct by considering \emph{operad algebras}, which are sets (or more generally, category objects) upon which an operad acts. To define this, we first define operad morphisms. 

\begin{definition}\label{def:opmorph}
Let $\underline{O} = (O,\tau,\id,\star),\underline{O}' = (O',\tau',\id,\star)$ be two operads. A \emph{morphism} $\varphi\colon \underline{O} \rightarrow \underline{O}'$ is a collection of maps $\varphi_n\colon O_n \rightarrow O'_n$ for each $n$, such that:
\begin{enumerate}
\item For all $n,m_1,\ldots,m_n \in \mathbb{Z}_{\geq 0}$ and $f \in O_n$, $g_i \in O_{m_i}$ one has
\begin{align*}
&\varphi_{\sum_i m_i}(f \star (g_1,\ldots,g_n))\\
&= \varphi_n(f) \star (\varphi_{m_1}(g),\ldots,\varphi_{m_n}(g));
\end{align*}
\item For each bijection $\sigma\colon [n] \iso [n]$ one has $\varphi_n \circ \tau_{\sigma} = \tau'_{\sigma} \circ \varphi_{n}$.
\end{enumerate}
\end{definition}

We omit the subscript $n$ from $\varphi_n(f)$ when $n$ is clear. As one expects, an operad morphism is a map between operads that preserves all relevant structure. An \emph{operad algebra} is then a morphism $\underline{O} \rightarrow \OEnd_{\mathcal{C}}(X)$.

\begin{definition}
Let $\CC$ be a category with finite products, and let $\underline{O}$ be an operad. An \emph{$\underline{O}$-algebra in $\CC$} is a pair $(X,\varphi)$ where $X$ is an object of $\CC$, and $\varphi\colon \underline{O} \rightarrow \OEnd_{\CC}(X)$ is an operad morphism.
\end{definition}

\begin{example}
Let $X$ be a convex subset of $\mathbb{R}^N$ for some $N$. To define a $\underline{\Delta}$-algebra structure on $X$ (in the category $\cat{Set}$), we need to define a map $\varphi(\vec{p})\colon X^n \rightarrow X$ for each $\vec{p} \in \Delta_n$. We do this by defining $\varphi(\vec{p})(x_1,\ldots,x_n) = \sum_{i=1}^n p_i x_i$, which can be shown to satisfy the axioms of Definition \ref{def:opmorph}. Note that if $X$ is a linear subspace of $\mathbb{R}^N$, then this actually defines a $\underline{\Delta}$-algebra in the category $\cat{Vec}_{\mathbb{R}}$ of real vector spaces, since each $\varphi(\vec{p})$ is a linear map.
\end{example}

\section{The attack tree operad}

In this section, we show how attack trees form an operad, and we show how AT metrics can be defined as algebras over this operad. We first define so-called anchored ATs, which are necessary to define the sets $\AT_n$ that make up the operad $\OAT$. 

\begin{definition} \label{def:anchor}
Let $T = (N,E,\gamma)$ and $T' = (N',E',\gamma')$ be ATs with $n$ BASs.
\begin{enumerate}
    \item An \emph{anchoring of $T$} is a bijection $\mu\colon B_T \iso [n]$. The pair $(T,\mu)$ is called an \emph{anchored AT}. 
    \item Let $\mu$ and $\mu'$ be anchorings of $T$ and $T'$. An \emph{anchor isomorphism $(T,\mu) \iso (T',\mu')$} is a graph isomorphism $m\colon N \iso N'$, such that $\gamma = \gamma' \circ m$ and $\mu = \mu' \circ m$.
    \item We call $(T,\mu)$, $(T',\mu')$ equivalent if there exists an anchor isomorphism $(T,\mu) \iso (T',\mu')$.
\end{enumerate}
\end{definition}

We are now in a position to define the operad of attack trees. The elements of $\AT_n$ are anchored ATs with $n$ BASs, and, as we will make more precise below, $\star$ is modular composition.

\begin{definition} \label{def:atoperad}
For $n \in \mathbb{Z}_{\geq 0}$, let $\AT_n$ be the set of equivalence classes of anchored ATs with $n$ BASs (so $\AT_0 = \varnothing$). Let $\recht{id} \in \AT_1$ be the AT consisting of a single BAS. Let $T = (N,E,\gamma,\mu) \in \AT_n$ and $T_i = (N_i,E_i,\gamma_i,\mu_i) \in \AT_{m_i}$ for each $i \leq n$. Assume $N$ and the $N_i$ are pairwise disjoint. We define $(N',E',\gamma',\mu') = T \star (T_1,\ldots,T_n) \in \AT_{\sum_i m_i}$ by
\begin{align*}
N' &= \{v \in N \mid \gamma(v) \neq \BAS\} \cup \bigcup_i N_i,\\
E' &= \ldb(v,w) \in E \mid \gamma(w) \neq \BAS\rdb\\
& \quad \quad \uplus \ldb(v,\R{T_i}) \mid (v,\mu^{-1}(i)) \in E\rdb \uplus \biguplus E_i,\\
\gamma'(v) &= \begin{cases}
\gamma(v), & \textrm{ if $v \in N$}\\
\gamma_i(v) & \textrm{ if $v \in N_i$}
\end{cases},\\
\mu'(v) &= \sum_{i' < i} m_i + \mu_i(v) \text{ if $v\in N_i$}.
\end{align*}
For each $\sigma\colon [n] \iso [n]$ define $\tau_{\sigma}(T,\mu) = (T,\sigma \circ \mu)$. Then these data together define the \emph{Attack tree operad} $\OAT$.
\end{definition}

Every equivalence class in $\AT_n$ has a representative for which $B_T = \{a_1,\ldots,a_n\}$ and $\mu(a_i) = i$. Unless specified otherwise, in what follows we always take such a representative, and we suppress $\mu$ in the notation.  The composition $T \star (T_1,\ldots,T_n)$ has an elaborate definition, but it is just the modular composition $T[T_1/a_1,\ldots,T_n/a_n]$, i.e., each BAS $a_i$ is replaced with the AT $T_i$.

The objects of $\OAT$ are isomorphism classes of anchored ATs, rather than (anchored) ATs themselves. This makes the definition more elaborate, but is necessary to express the fact that two ATs with the same gates and topological structure should have the same metrics.

We are now in a position to define AT metrics. Traditionally, a metric assigns to each $T \in \AT_n$, and each collection of BAS values $(x_i)_{i \leq n}$, where each $x_i$ is in a given set $X$, a value $x_T$ representing which is again in $X$. Shifting the perspective slightly, a metric assigns to each $T \in \AT_n$ a map $X^n \rightarrow X$, whose input is $(x_n)_{i \leq n}$ and whose output is $x_T$. In other words, we have a map $\AT_n \rightarrow \End_n(X)$, which then directly leads to the following definition:

\begin{definition}[AT metric] \label{def:atmetric}
Let $\CC$ be a category with finite products. An \emph{attack tree metric} in $\CC$ is an $\OAT$-algebra in $\CC$.
\end{definition}

The fact that an $\OAT$-algebra arises from an operad morphism $\varphi\colon\OAT \rightarrow \OEnd_{\CC}(X)$ has an important implication: The preservation of $\star$ means that the metric allows for \emph{modular analysis}. To calculate the metric on $T[T'/a]$ given the $x_i$ for each BAS $a_i$, first calculate $x_{T'}$ using $\varphi(T')$, and then take $x_a = x_{T'}$ as an imput for $\varphi(T)$ to calculate $x_T$. Many existing AT metrics have been shown to satisfy modular analysis \cite{mauw2005foundations,lopuhaa2021attack}; furthermore, modular analysis fits within the philosophy that an AT can be refined by replacing a BAS with a new AT without affecting quantitative analysis. In the next section, we will give necessary and sufficient conditions for AT metrics to fall under Definition \ref{def:atmetric}.

\begin{remark}
We allow AT metrics to live in another category than sets. This is mainly useful for stating that a metric is structure-preserving in some way. For instance, the fact that for every $T$ the map $\vec{x} \mapsto x_T$ is continuous (order-preserving, linear,...) means that the metric lives in the category $\cat{Top}$ ($\cat{Pos},\cat{Vec}_{\mathbb{R}},\ldots$).
\end{remark}

\begin{example} \label{exa:metric} \begin{enumerate}
    \item Consider the structure function $\struc{T}$ from Definition \ref{def:struc}; we show how this can be interpreted as an AT metric in the category $\cat{Set}$. To do this, we take $X = \mathbb{B}$, and we have to define a map $\phstr\colon \AT_n \rightarrow \End_n(\BB)$ for every $n \geq 1$. This means that for every AT $T$ with $n$ BASs, we have to define a map $\phstr(T)\colon \BB^n \rightarrow \BB$. The structure function can be viewed as such, by taking $\phstr(T)(\vec{b}) = \struc{T}(\vec{b},\R{T})$ for all $\vec{b} \in \BB^n$. A straightforward, but technical proof shows that this $\phstr$ satisfies Definition \ref{def:opmorph}, so $(\BB,\phstr)$ is an AT metric.
    
    Since ATs do not contain $\mathtt{NOT}$-gates, the function $\vec{b} \mapsto \struc{T}(\vec{b},\R{T})$ is monotonous. This means that we can regard $\phstr$ as a map $\OAT \rightarrow \OEnd_{\cat{Pos}}(\BB)$, so $(\BB,\phstr)$ is an $\OAT$-algebra in $\cat{Pos}$. We denote $\OEnd_{\cat{Pos}}(\BB) =: \OMBool$, the operad of monotonous Boolean functions.
    \item Give a probability $p_i \in [0,1]$ to each BAS $a_i$. The  \emph{total attack probability} is
\begin{equation*}
\recht{TAP}(T,\vec{p}) = \sum_{\vec{b} \in \recht{Suc}_T} \prod_{i=1}^n  p_i^{b_i}(1-p_i)^{1-b_i} \in [0,1],
\end{equation*}
i.e., $\recht{TAP}(T,\vec{p})$ is the probability that an attack is succesful given the compromise probabilities of all BASs (recall that $\recht{Suc}_T$ is the set of all successful attacks). Then $\varphi^{\recht{TAP}}\colon \OAT \rightarrow \End([0,1])$ given by $\varphi^{\recht{TAP}}(T)(\vec{p}) = \recht{TAP}(T,\vec{p})$ is an operad morphism, and so total attack probability is an AT metric (in the category of topological posets, in fact). It cannot be written as one of the two semiring metrics below.
    \item For an AT $T$, let $M_T$ be the set of minimal elements in $\recht{Suc}_T$, for the partial order inherited from $\BB^n$; this is the set of minimal succesful attacks, usually called \emph{minimal attacks}. Furthermore, let $D = (X,\triangle,\triangledown)$ be a semiring, i.e., $X$ is a set, $\triangle$ and $\triangledown$ are binary commutative associative operations on $X$, and $\triangledown$ distributes over $\triangle$. For $\vec{x} \in X^n$, the security metric $\hat{x}(T)$ is defined as
\begin{equation} \label{eq:metCarlos}
\hat{x}(T) = \bigtriangledown_{\vec{b} \in M_T} \bigtriangleup_{i\colon b_i = 1} x_i.
\end{equation}
These are the \emph{propositional semiring metrics} as defined in \cite{9925106}. For example, if $D = (\mathbb{Z}_{\geq 0},\min,+)$, then $\hat{x}$ is the \emph{minimal cost} metric: the minimal cost an attacker needs to succesfully attack the system.

Define a map $\varphi^{\recht{ps}}_{D,n}\colon \AT_n \rightarrow \End_n(V)$ by $\varphi^{\recht{ps}}_{D,n}(T)(\vec{x}) = \hat{x}(T)$; this defines an $\OAT$-algebra $(X,\varphi^{\recht{ps}}_{D})$ in $\cat{Set}$. To show that $\varphi^{\recht{ps}}$ is an operad morphism, i.e., that it preserves $\star$, is quite non-trivial; compare Theorem 9.2 of \cite{9925106}. Changing $M_T$ by other sets of attacks in \eqref{eq:metCarlos} yields different, noncompatible metrics, see for instance the so-called set semantics of \cite{bossuat2017evil}.
\item Let $T \in \AT_n$ and let $x \in X^n$, for a semiring $D = (X,\triangle,\triangledown)$. we extend $x$ recursively to a $\tilde{x} \in V^N$ by
\begin{equation*}
\tilde{x}_v = \begin{cases}
x_i, & \textrm{ if } v = b_i,\\
\bigtriangledown_{w \in \ch(v)} \tilde{x}_w, & \textrm{ if } \gamma(v) = \tOR,\\
\bigtriangleup_{w \in \ch(v)} \tilde{x}_w, & \textrm{ if } \gamma(v) = \tAND.
\end{cases}
\end{equation*}
The security metric for $T$ is then defined $\tilde{x}(T) = \tilde{x}_{\R{T}}$. The map $\varphi_{D,n}^{\recht{bu}}\colon \AT_n \rightarrow \End_n(V)$ given by $\varphi_{D,n}^{\recht{bu}}(T)(\vec{x}) = \tilde{x}(T)$ defines an $\OAT$-algebra $(X,\varphi_D^{\recht{bu}})$ in $\cat{Set}$. Note that this is a different $\OAT$-algebra structure than the one in point 3) above; these are the \emph{bottom-up semiring metrics} of \cite{mauw2005foundations}.
\end{enumerate}
\end{example}

\begin{remark}
From the example above it is clear that even ostensibly straightforward metrics like \emph{min cost} have multiple, noncompatible definitions, even though some of those are based on the same mathematical formalism (semirings). This might seem strange to the reader; in my opinion there are two main causes for this. The first is that there are different ideas on what a node $v$ with multiple parents in an AT represents. In some works, this means that $v$ must be activated multiple times, once for every parent \cite{mauw2005foundations}. In other works, activating $v$ once is enough to serve as input for all its parents \cite{budde2021efficient}. The first interpretation naturally incurs a greater overall attack cost; these metrics also have semantics slightly different from Definition \ref{def:semantics}, phrased in terms of multisets instead.

The second reason for noncompatible definitions is that metrics are often defined alongside the algorithms to compute them; as such, the metric is chosen in such a way that it conforms to the algorithm. We discuss this in detail in Section \ref{sec:ext} for dynamic ATs, where \emph{min time} has multiple definitions, and each of them conforms to the algorithm introduced to calculate it.
\end{remark}

\begin{remark}
In some works, treelike structures in computer science are modeled not by operads but by so-called \emph{multicategories} a.k.a. \emph{coloured operads} \cite{meseguer1989general}. Where operads axiomatize maps $X^n \rightarrow X$, multicategories axiomatize maps $X_1 \times \ldots \times X_n \rightarrow Y$. For AT metrics this higher level of generalization is not necesary, since each BAS' security value lies in the same domain, and the AT's security value lies in this domain as well; for example $X = \mathbb{R}_{\geq 0}$ for the minimal attack time metric. Metrics where this is not the case are discussed in Section \ref{sec:ubi}. In Section \ref{sec:ext}, we show how metrics on attack-defense trees can be modeled as bicoloured operads.
\end{remark}

\section{Necessary and sufficient conditions for operad metrics} \label{sec:ubi}

We argue that almost all metrics that one encounters are, in fact, operad metrics. We show this by giving necessary and sufficient conditions for a metric to be an operad algebra, and argue that these conditions are natural. This section is, in effect, an informal summary of the preceding section.

Setting aside Definition \ref{def:atmetric}, any attack tree metric from the literature can be phrased as assigning to each attack tree $T$ and each vector of $n$ BAS attribute values $\vec{x}$, a security value $F(T,\vec{x})$. To encode this as an operad algebra, one needs to define the function $\varphi(T)$ as $\varphi(T)(\vec{x}) = F(T,\vec{x})$. This is not so much an encoding task as it is a shift of perspective, i.e., an AT metric is not a \emph{value} but a \emph{function}. In order for $\varphi$ to be an operad algebra, i.e., to fit Definition \ref{def:atmetric}, it is necessary and sufficient to satisfy the following conditions:
\begin{enumerate}
\item \emph{Isomorphism independence}: two ATs that are the same, except for the names of the nodes, should yield the same metric value.
\item \emph{Input-output correspondence}: the set of possible metric values must be a subset of the set of possible BAS attribute values.
\item \emph{Allows for modular analysis}: If $T' = T \star (T_1,\ldots,T_n)$, then $\varphi(T')(\vec{x}) = \varphi(T)(\vec{y})$, where $y_i = \varphi(T_i)(x^{(i)})$ and $x^{(i)}$ is the restriction of $\vec{x}$ to the BASs of $T_i$.
\end{enumerate}

Condition 1) is true for any reasonable metric, since metric values should only depend on the mathematical structure. In the operad algebra framework, this is reflected by the fact that $\AT_n$ consists of isomorphism classes of ATs rather than ATs themselves. Some metrics need more information than the AT structure, and this can be done by extending the mathematical framework, as we do for dynamic ATs in Section \ref{sec:ext}.

Condition 2) is usually met: for instance, each BAS is assigned a cost value (in $\mathbb{R}_{\geq 0}$) used to calculate the minimal cost of a succesful attack (again in $\mathbb{R}_{\geq 0}$). However, there are counterexamples: for instance, each BAS can be assigned a compromise time, given as a probability distribution on $\mathbb{R}_{\geq 0}$,  and the desired metric might be the \emph{mean time to compromise} of the entire AT, which is an element of $\mathbb{R}_{\geq 0}$ \cite{mcqueen2006time}.

In such cases, it is often possible to put the metric in the operad algebra framework by retaining additional structure. For instance, the \emph{mean time to compromise} metric above can be extended so that the metric value is not just the expected value of compromise time, but the entire probability distribution of compromise time; then both the BAS values and the metric values are probability distributions on $\mathbb{R}_{\geq 0}$. Retaining this additional structure does have the downside of potentially making metric calculation more complicated.

Condition 3) is the crux of Definition \ref{def:atmetric}, as this ensures that $\varphi$ respects the operad operation $\star$. This is almost always satisfied: while it is possible to create pathological examples of metrics that do not meet condition 3), like Example \ref{ex:path} below, all metrics from the literature known to me that satisfy 2) also satisfy 3). Modular analysis is natural from the hierarchical nature of attack trees, in which an AT can be further refined by replacing a BAS with a subtree describing its vulnerabilities. Condition 3) then states that the metric respects this refinement: calculating the metric for the large AT should yield the same result as first calculating the metric value of the subtree, and then using this as an input for further calculation. Since AT refinement is such a natural operation, it is no surprise that many metrics respect this refinement. This also shows why operad algebras are such a natural fit for AT metric formalization, as the operad axioms capture the one property that is, in our mind, fundamental to the nature of AT metrics. An example of a metric for attack-defense trees that does \emph{not} allow for modular analysis is given in Section \ref{sec:ext}.

Although condition 3) is very natural for AT metrics, it can nonetheless be non-straightforward to show that a metric satisfies it. When a metric is defined bottom-up, like Examples \ref{exa:metric}.1) and \ref{exa:metric}.4), or top-down, condition 3) is usually met automatically. However, when a metric is defined taking the entire AT structure into account, like \ref{exa:metric}.2) and \ref{exa:metric}.3), a proof of the applicability of modular analysis can be quite nontrivial \cite{lopuhaa2021attack}.

\begin{example} \label{ex:path}
Consider the \emph{size} metric of an AT, that assigns to an AT $T = (N,E,\gamma)$ the integer $\varphi(T) = |N|$. We can consider this as a function $\mathbb{N}^n \rightarrow \mathbb{N}$ by taking it to be the constant function with value $|N|$. Then this does not satisfy modular analysis: If $T' = (N',E',\gamma')$ satisfies $T' = T \star (T_1,\ldots,T_n)$, then $\varphi(T')(\vec{x}) = |N'|$ regardless of $\vec{x}$, while $\varphi(T)(\vec{y}) = |N|$.
\end{example}

\section{Extensions of the AT framework} \label{sec:ext}

% \begin{remark} \label{rmk:ext}
% Most extensions of the AT framework work by adding additional gate types, such as sequential $\tAND$-gates \cite{jhawar2015attack} or defense nodes \cite{kordy2018quantitative}. Such \emph{dynamic ATs} or \emph{attack-defense trees} also form an operad, which we obtain by extending the codomain of $\gamma$ in Definition \ref{def:at} (along with some possible extra decoration, such as an ordering of the children of a sequential $\tAND$-gate) and taking this as a basis for Definition \ref{def:atoperad}. A metric for such extended ATs is then an algebra over this extended operad; the dynamic AT metrics of \cite{9925106} and the attack-defense tree metrics of \cite{kordy2018quantitative} are examples of this. Thus our framework for AT metrics can easily be extended to extensions of ATs. 
% \end{remark}

In this section, we showcase how the operad algebra framework can be extended to define metrics for extensions of ATs. We consider two popular extensions, dynamic ATs and attack-defense trees, and their metrics. The general technique, however, has wider applicability: to define metrics on an extended AT formalism, first define the operad of such extended ATs. Then a metric is simply an algebra over that operad.

\subsection{Dynamic ATs}

\begin{figure}
\centering
\begin{subfigure}{0.49\linewidth}
\centering
\begin{tikzpicture}[
and/.style={and gate US,rotate=90,draw,fill = lightgray},
or/.style={or gate US,rotate=90,draw,fill = lightgray},
bas/.style={circle,draw,fill = lightgray}]
\draw (1.5,-1) node [bas] (b) {\small $a$};
\draw (2.5,-1) node [bas] (c) {\small $b$};
\draw (2,0) node [and] (a1) {};
\draw (3,0) node [bas] (a2) {$c$};
\draw (2.5,1) node [and] (o) {\rotatebox{270}{\small $\rightarrow$}};
\draw (b) -- (a1) -- (o) -- (a2);
\draw (a1) -- (c);
\end{tikzpicture}
\caption{A dynamic AT. In order to reach the root, $a$ and $b$ can be activated in parallel, and once both are done $c$ can be activated. The total attack time is $\max(t_a,t_b)+t_c$.} \label{fig:DAT}
\end{subfigure}
\begin{subfigure}{0.49\linewidth}
\centering
\begin{tikzpicture}[
and/.style={and gate US,rotate=90,draw,fill = lightgray},
or/.style={or gate US,rotate=90,draw,fill = lightgray},
bas/.style={circle,draw,fill = lightgray},
co/.style={diamond,draw,fill=lightgray}]
\draw (0.5,0) node [or,fill=red!20] (rt) {};
\draw (-0.5,-1) node[co,fill=red!20] (c1) {};
\draw (1.5,-1) node[bas,fill=red!20] (b1) {$d$};
\draw (-0.5,-2) node[bas,fill=red!20] (b2) {$a$};
\draw (0.5,-2) node[co,fill=green!20] (c2) {};
\draw (0.5,-3) node[bas,fill=green!20] (b3) {$b$};
\draw (1.5,-3) node[bas,fill=red!20] (b4) {$c$};
\draw (b4) -- (1.5,-2) -- (c2) -- (0.5,-1) -- (c1) -- (rt) -- (b1);
\draw (b3) -- (c2);
\draw (b2) -- (c1);
\end{tikzpicture}
\caption{An attack-defense tree. The attacker (red) can reach the root by attack $a$; this can be prevented by the defender's (green) action $b$. This can itself be stopped by attacker action $c$. Alternatively, the attacker can reach the root by attack $d$ which cannot be stopped by the defender.} \label{fig:ADT}
\end{subfigure}
\caption{A dynamic AT and an attack-defense tree.}
\end{figure}
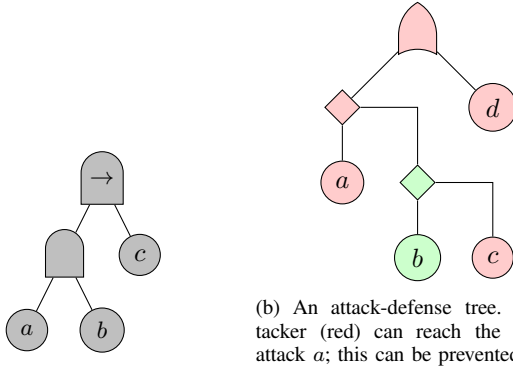

Dynamic ATs (DATs) are used when not only the set of BAS undertaken by the attacker are relevant, but also their relative order \cite{jhawar2015attack}. DATs have a new gate type, $\tSAND$ (``sequential $\tAND$''), where $\tSAND(v_1,v_2)$ is only activated if both $v_1$ and $v_2$ are activated, and $v_1$ is activated before $v_2$ (see Figure \ref{fig:DAT}). Thus a $\tSAND$-gate also needs to know the relative order of its children, leading to the following definition.

\begin{definition} \label{def:DAT}
A \emph{dynamic attack tree} is a triple $T = (N,E,\gamma,\alpha)$, where $(N,E)$ is a rooted directed acyclic multigraph and $\gamma$ is a function $\gamma\colon N \rightarrow \{\tOR,\tAND,\BAS,\tSAND\}$ such that $\gamma(v) = \BAS$ if and only if $v$ is a leaf of $(N,E)$. Furthermore, $\alpha$ consists of, for every $v \in N$ with $\gamma(v) = \tSAND$ with $n$ children, a bijection $\alpha_v\colon [n] \iso \ch(v)$.
\end{definition}

Informally, modular composition of DATs is exactly the same as that of ATs, by replacing a BAS with another AT. This can be formalized analogous to Definition \ref{def:mod}. One can also define anchorings and isomorphisms of DATs analogous to Definition \ref{def:anchor}, leading to the definition of the operad $\ODAT$ of DATs. Then DAT metrics are defined as follows.

\begin{definition}
Let $\mathcal{C}$ be a category with finite products. A \emph{DAT} metric in $\mathcal{C}$ is an $\ODAT$-algebra in $\mathcal{C}$.
\end{definition}

\begin{example} The metrics below are all noncompatible, and yield the four definitions of \emph{min time} of Figure 1.
\begin{enumerate}
    \item One way to define the semantics on a DAT $B_T$  is by defining an attack to be a tuple $(A,\prec)$, where $A \subseteq B_T$ and $\prec$ is a partial order on $A$; the intuition is that $a \prec b$ means that $a$ must be finished before $b$ is started \cite{budde2021efficient,lopuhaa2021attack}. If each BAS is assigned a duration $x_a \in \mathbb{R}_{\geq 0}$, then the total duration of an attack $(A,\prec)$ is equal to
    \[
    \recht{t}(A,\prec) = \max_{\substack{C \recht{ max. chain}\\ \recht{of 
 } (A,\prec)}} \sum_{a \in C} x_a.
    \]
    Here $C$ ranges over the maximal chains of $A$, i.e., the maximal subsets on which $\prec$ is a linear order; all BASs in such a chain must happen sequentially, which explains the summation. Once we define the set of minimal attacks $M_T$ on $T$, we then define the \emph{minimal attack time} as
    \[
    \recht{mt}(T) = \min_{(A,\prec) \in M_T} \recht{t}(A,\prec).
    \]
    There are multiple, slightly conflicting definitions of $M_T$, leading to different metrics \cite{budde2021efficient,lopuhaa2021attack}. Both of these choices can be shown to satisfy modular analysis and hence can be viewed as DAT metrics $\ODAT \rightarrow \OEnd(\mathbb{R}_{\geq 0})$. This definition of \emph{min time} can also be extended to more general \emph{dynamic semiring metrics} \cite{budde2021efficient} which extend the semiring framework of Example \ref{exa:metric}.3) with a third operator $\lozenge$.
    \item Given such a dynamic semiring $(X,\triangledown,\triangle,\lozenge)$, we can extend the bottom-up metric of Example \ref{exa:metric}.4) with $\tilde{x}_v  = \lozenge_{w \in \ch(v)} \tilde{x}_w$ if $\gamma(v) = \tSAND$, yielding the \emph{bottom-up dynamic semiring metrics} of \cite{jhawar2015attack}. For \emph{min time}, one has $\triangledown = \min$, $\triangle = \max$, $\lozenge = {+}$. This metric can also be expressed in terms of series-parallel graph semantics. As bottom-up methods, it is straightforward to show that modular analysis applies, and that this can be expressed as a $\ODAT$-algebra.
    \item A top-down approach to DAT metrics based on priced timed automata (PTA) is developed in \cite{kumar2015quantitative}. A PTA, constructed in UPPAAL from the DAT, can be used to calculate metrics such as time and cost. Again one can show that these metrics can be regarded as $\ODAT$-algebras, as this is a straightforward consequence from its top-down definition. However, this approach does not consider the feasibility of attacks, yielding different values of \emph{min time} from the definitions above.
\end{enumerate}
\end{example}

\subsection{Attack-Defense Trees}

Attack-defense trees (ADTs) were introduced to not only model attacker or \emph{proponent} actions and goals, but also defender or \emph{opponent} actions and goals. To this end, each node $v$ is assigned a label $\varrho(v)$, either $\pp$ or $\oo$, to mark it as an opponent or proponent action or goal. $\tOR$- and $\tAND$-gates marked with an actor can only have children belonging to the same actor. A new ``counter'' gate $\tC$ is added, that must have exactly two children with different labels: if $v = \tC(v_1,v_2)$ with $\varrho(v) = \varrho(v_1) \neq \varrho(v_2)$, then $v$ is activated if and only if $v_1$ is activated, but $v_2$ is not. This allows a defender to stop attack $v_1$ from propagating, by activating defense $v_2$ (see Figure \ref{fig:ADT}). Formally, ADTs are defined as follows.

\begin{definition}
A \emph{attack-defense tree} is a quadruple $T = (N,E,\gamma,\varrho)$ where $(N,E)$ is a rooted directed acyclic graph and $\gamma\colon N \rightarrow \{\tOR,\tAND,\BAS,\tC\}$ and $\varrho\colon N \rightarrow \{\pp,\oo\}$ are functions satisfying:
\begin{enumerate}
    \item $\gamma(v) = \BAS$ if and only if $v$ is a leaf;
    \item if $v$ is a nonleaf with $\ch(v) = v_1,\ldots,v_n$, then:
    \begin{enumerate}
    \item if $\gamma(v) \in \{\tOR,\tAND\}$, then $\varrho(v_i) = \varrho(v)$ ;
    \item if $\gamma(v) = \tC$, then $n = 2$ and $\varrho(v_1) \neq \varrho(v_2)$.
    \end{enumerate}
\end{enumerate}
\end{definition}

Unlike ATs and DATs, modular composition for ADTs is more nuanced: a BAS can only be replaced by an ADT whose top node agent ($\varrho$-value) is the same as the BAS agent. This can be phrased in terms of \emph{bicoloured operads}, an operad in which only `matching' elements can be composed. Below, we briefly sketch the idea behind coloured operads; for more information the reader is referred to \cite{yau2016colored}.

Instead of $\OEnd(X)$, we consider two sets (or in general, objects in a category) $X_{\pp}$, $X_{\oo}$, and maps of the form $X_{\pp}^n \times X_{\oo}^m \rightarrow X_{\pp}$ and $X_{\pp}^n \times X_{\oo}^m \rightarrow X_{\oo}$. These form sets $\Bic_{n,m,\pp}(X_{\pp},X_{\oo})$ and $\Bic_{n,m,\oo}(X_{\pp},X_{\oo})$. On these sets, one has a composition operation $\star$ that is defined only on matching domains and codomains. Taken all together, this is the structure $\OBic(X_{\pp},X_{\oo})$; axiomatising such a structure leads to the concept of a bicoloured operad. Defining composition as usual, we can deifne the bicoloured operad $\OADT$ of attack-defense trees. This leads to the following definitions:

\begin{definition}
Let $\mathcal{C}$ be a category with finite products.
\begin{enumerate}
    \item Let $\underline{B}$ be a bicoloured operad. A $\underline{B}$-\emph{algebra} in $\mathcal{C}$ is triple $(X_{\pp},X_{\oo},\varphi)$ where $X_{\pp},X_{\oo}$ are objects in $\CC$ and $\varphi\colon \underline{B} \rightarrow \OBic(X_{\pp},X_{\oo})$ is a morphism of bicoloured operads.
    \item An \emph{ADT metric} in $\mathcal{C}$ is an $\OADT$-algebra in $\mathcal{C}$.
\end{enumerate}
\end{definition}

\begin{example}
Consider a tuple $(X,\triangledown^{\pp},\triangle^{\pp},\square^{\pp},\triangledown^{\oo},\triangle^{\oo},\square^{\oo})$, where $X$ is a set and the last six elements are binary operations on $X$. Via such a set two ADT metrics are defined in \cite{kordy2018quantitative}. First, one can consider a bottom-up definition similar to Example \ref{exa:metric}.4), where the operator of $v$ depends on both $\gamma(v)$ and $\varrho(v)$ ($\square^{\pp},\square^{\oo}$ are for $\tC$-gates). For instance, the minimal cost of an attack that succeeds regardless of defender actions corresponds to the tuple $([0,\infty],\min,+,+,+,\min,\min)$, with each $\oo$-BAS having attribute value $\infty$. As a bottom-up metric, it is straightforwardly shown to satisfy modular analysis, and it can be expressed as an $\OADT$-algebra with $X_{\oo} = X_{\pp} = X$.

Alternatively, similar to Example \ref{exa:metric}.3) one can first define a set of succesful attacks, and then define the metric globally analogous to \eqref{eq:metCarlos}. This is done via so-called \emph{set semantics} in \cite{kordy2018quantitative}. However, this does not generally yield a metric in the operad-theoretic sense, as this definition is not compatible with modular decomposition; for that one needs additional assumptions on the operations and their interplay. Precisely defining the required axioms is beyond the scope of this paper.
\end{example}

\section{Bottom-up algorithm for metric computation} \label{sec:bu}

By design, our AT metrics are considerably more general than other definitions in the literature: we wish to encapsule many existing definitions into one common framework. A consequence is, however, that existing algorithms to calculate AT metrics cannot be shown to work in general under our definition. In this section and the following one, we generalize existing metric calculation algorithms to the operad algebra setting, and give sufficient operad-theoretic properties, such that the algorithms can be applied to metrics that satisfy these properties. This section is dedicated to the bottom-up algorithm first defined in \cite{mauw2005foundations}. For the metrics it applies to it is the state-of-the-art method with only linear complexity; thus the increased generality of the operad framework does not come at a loss of efficiency.  Throughout this section, we consider metrics in the category $\cat{Set}$.

\subsection{The operators $C_n$ and $D_n$}

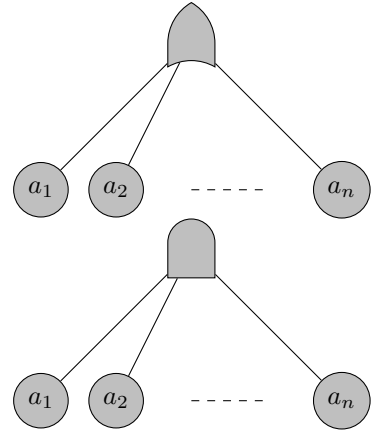
\begin{figure}
\centering
\begin{tikzpicture}[
and/.style={and gate US,rotate=90,draw,fill = lightgray},
or/.style={or gate US,rotate=90,draw,fill = lightgray},
bas/.style={circle,draw,fill = lightgray}]
\draw (0,1) node[or] (a) {{\color{lightgray}$x_n$}};
\draw (-2,-1) node[bas] (b1) {$a_1$};
\draw (-1,-1) node[bas] (b2) {$a_2$};
\draw (2,-1) node[bas] (bn) {$a_n$};
\draw (b1) -- (a);
\draw (b2) -- (a);
\draw (bn) -- (a);
\draw[dashed] (0,-1) -- (1,-1);
\end{tikzpicture}

\begin{tikzpicture}[
and/.style={and gate US,rotate=90,draw,fill = lightgray},
or/.style={or gate US,rotate=90,draw,fill = lightgray},
bas/.style={circle,draw,fill = lightgray}]
\draw (0,1) node[and] (a) {{\color{lightgray}$x_n$}};
\draw (-2,-1) node[bas] (b1) {$a_1$};
\draw (-1,-1) node[bas] (b2) {$a_2$};
\draw (2,-1) node[bas] (bn) {$a_n$};
\draw (b1) -- (a);
\draw (b2) -- (a);
\draw (bn) -- (a);
\draw[dashed] (0,-1) -- (1,-1);
\end{tikzpicture}
\caption{The ATs $T_{\tOR,n}$ and $T_{\tAND,n}$.} \label{fig:torand}
\end{figure}

For $n \in \mathbb{Z}_{\geq 1}$, let $T_{\tOR,n} \in \AT_n$ be the AT with $n$ BASs and a single non-BAS (the root), which is an $\tOR$-gate, see Figure \ref{fig:torand}. For an $\OAT$-algebra $\varphi\colon \OAT \rightarrow \OEnd(X)$ in $\cat{Set}$, let $C_n\colon X^n \rightarrow X$ be the map $\varphi(T_{\tOR,n})$. Because $\varphi$ is an operad morphism, the operator $C_n$ does not depend on the order of its arguments, since the BAS of $T_{\tOR,n}$ can be permuted freely without changing its isomorphism class of anchored AT. Therefore, $C_n$ is symmetric. Similarly, there is a symmetric operator $D_n\colon X^n \rightarrow X$ coming from the analogously-defined AT $T_{\tAND,n}$. The operators $C_n$ and $D_n$ depend on the operad algebra $(X,\varphi)$, but we will omit this from the notation.

\begin{example}
Consider the bottom-up metric from Example \ref{exa:metric}.4). It is easy to see that here $C_n(\vec{x}) = \bigtriangledown_{i=1}^n x_i$, and $D_n(\vec{x}) = \bigtriangleup_{i=1}^n x_i$. Less obvious is that the metric from Example \ref{exa:metric}.3) has the same $C_n,D_n$.
\end{example}

In the example above $C_n$ and $C_m$ for $n \neq m$ are clearly related: they are the same operation applied to more arguments. This is not true in general: in our definition there need not be any relation between $C_n$ and $C_m$. In fact, any choice of symmetric operators $C_n$, $D_n$ can be extended to an AT metric.

\begin{algorithm}[t]
	\KwIn{$T = (N,E,\gamma) \in \AT_n$; $\varphi\colon \OAT \rightarrow \OEnd(X)$; $\vec{x} \in X^n$; $v \in N$.}
	\KwOut{$\BU(T,\varphi,\vec{x},v) = \varphi(T_v)(\vec{x}_v) \in X$}
	\BlankLine
	\uIf{$v = a_i$}{%
		\Return{$x_i$}
	} \uElseIf{$v = \tOR(w_1,\ldots,w_m)$}{%
		$\forall i \leq m\colon y_i \leftarrow \BU(T,\varphi,\vec{x},w_i)$\;
\Return{$C_n(y_1,\ldots,y_n)$}
	} \Else { \tcp{Now $v = \tAND(w_1,\ldots,w_m)$}%
 $\forall i \leq m\colon y_i \leftarrow \BU(T,\varphi,\vec{x},w_i)$\;
\Return{$D_n(y_1,\ldots,y_n)$}}
	\caption{The algorithm $\BU$. Recall that $T_v$ is the subDAG of $T$ with root $v$, so to compute $\varphi(T)(\vec{x})$ one takes $v = \R{T}$.}
	\label{alg:BU}
\end{algorithm}

\subsection{The bottom-up algorithm}

A `naive' algorithm to calculate the value AT-metric bottom-up is presented in Algorithm \ref{alg:BU}. 
In this algorithm, we write $\vec{x}_{w}$ for the vector in $X^{|B_{T_{w}}|}$ consisting of all $x_i$ such that the BAS $a_i$ is part of the AT $T_{w}$ (recall that $T_w$ is the subDAG of $T$ with root $w$). Thus $\BU(T,\varphi,\vec{x},v)$ calculates the metric value of the AT $T_v$, for BAS metric values inherited from $(T,\vec{x})$. In order to get $\varphi(T)(\vec{x})$, we take $v = \R{T}$. The algorithm works by applying, at every node $v$, either $C_n$ or $D_n$ to the values of its children, all the way up to the root, whose value is the AT metric. For specific metrics, this algorithm has been discussed many times in the existing literature, see for instance \cite{mauw2005foundations,kordy2010foundations,9925106}. Algorithm \ref{alg:BU} is concise and efficient: it has time complexity $\mathcal{O}(|N|+|E|)$. However, its main drawback is that it does not calculate every AT metric correctly \cite{9925106}. Nevertheless, Algorithm \ref{alg:BU} works when when the underlying graph is an actual tree and not just any DAG.

\begin{definition}
An AT $T = (N,E,\gamma)$ is called \emph{treelike} if the underlying graph $(N,E)$ is a (rooted) tree.
\end{definition}

\begin{theorem} \label{thm:BU}
Let $T \in \AT_n$, let $(X,\varphi)$ be an $\OAT$-algebra, and let $\vec{x} \in X^n$. If $T$ is treelike, then $\varphi(T)(\vec{x}) = \BU(T,\varphi,\vec{x},\R{T})$.
\end{theorem}

Theorem \ref{thm:BU} does not hold for DAG-shaped ATs, but it has been shown to hold for a number of metrics \cite{kordy2018quantitative,9925106}. For the remainder of this section, we give a necessary and sufficient, operad-theoretic condition under which Theorem \ref{thm:BU} also holds for DAG-shaped ATs. We start by developing the necessary machinery.

\subsection{Scoperads}

\begin{wrapfigure}[8]{r}{3cm}
\centering
\begin{tikzpicture}[
and/.style={and gate US,rotate=90,draw,fill = lightgray},
or/.style={or gate US,rotate=90,draw,fill = lightgray},
bas/.style={circle,draw,fill = lightgray}]
\draw (0,0) node[and] (r) {\rotatebox{270}{}};
\draw (-0.5,-1) node[or] (f) {\rotatebox{270}{}};
\draw (0.5,-1) node[or] (s) {\rotatebox{270}{}};
\draw (0,-2) node[bas] (b) {$b$};
\draw (1,-2) node[bas] (c) {$c$};
\draw (-1,-2) node[bas] (a) {$a$};
\draw (f) -- (r) -- (s) -- (b);
\draw (s) -- (c);
\draw (a) -- (f) -- (b);
\end{tikzpicture}
\end{wrapfigure}
The essential ingredient of the proof of Theorem \ref{thm:BU} is that every treelike $T$ can be modularly decomposed, via $\star$, into basic ATs of the forms $T_{\tOR,n},T_{\tAND,n}$. As this does not hold for general, DAG-shaped $T$, $\BU$ does not correctly calculate AT metrics for every $T$. From a graph-theoretical perspective every AT is built up from $\tOR$- and $\tAND$-gates. However, the operad structure of $\OAT$, i.e. modular composition, does not `see' this: ATs such as $\tAND(\tOR(a,b),\tOR(b,c))$ are indecomposable from an operad perspective. To amend this discrepancy, we introduce \emph{scoperads} (``surjectively complete operads''): operads with additional structure, that allows us to decompose any AT into $\tOR/\tAND$-gates on an operad level. The necessary structure is the existence of $\tau_{\sigma}$ for all surjective $\sigma$: as we will show below, this allows us to merge BASs in $\OAT$.

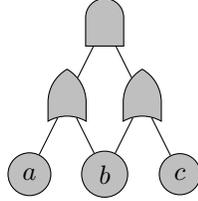
\begin{definition}
A \emph{scoperad} is a tuple $\underline{O} = (O,\tau,\id,\star)$ of a set $O_n$ for each $n \in \mathbb{Z}_{\geq 0}$; a \emph{map} $\tau_{\sigma}\colon O_{n} \rightarrow O_{n'}$ for each \underline{surjective} map $\sigma\colon [n] \twoheadrightarrow [n']$ for $n,n' \in \mathbb{Z}_{\geq 0}$; an element $1 \in O_1$; and a map $\star \colon O_n \times \prod_{i=1}^n O_{m_i} \rightarrow O_{\sum_i m_i}$ for all $n,m_1,\ldots,m_n \in \mathbb{Z}_{\geq 0}$. These must satisfy 1)--4) of Definition \ref{def:operad} for all $\sigma$, and 5) for bijective $\sigma$.
\end{definition}

\begin{figure}
\centering
\begin{tikzpicture}[
and/.style={and gate US,rotate=90,draw,fill = lightgray},
or/.style={or gate US,rotate=90,draw,fill = lightgray},
bas/.style={circle,draw,fill = lightgray}]
\draw (0,0) node[and] (r) {\rotatebox{270}{}};
\draw (-0.5,-1) node[or] (f) {\rotatebox{270}{}};
\draw (0.5,-1) node[bas] (s) {$a_3$};
\draw (0,-2) node[bas] (b) {$a_2$};
\draw (-1,-2) node[bas] (a) {$a_1$};
\draw (a) -- (f) -- (r) -- (s);
\draw (f) -- (b);
\draw (0,-3) node {$T$};

\draw[->] (2,-0.5) -- (2,-1.5);
\draw[->] (2.5,-0.5) -- (2.5,-1.5);
\draw[->] (3,-0.5) -- (2.5,-1.5);
\draw (2,-0.2) node {$1$};
\draw (2.5,-0.2) node {$2$};
\draw (3,-0.2) node {$3$};
\draw (2,-1.8) node {$1$};
\draw (2.5,-1.8) node {$2$};
\draw (2.5,-3) node {$\sigma$};

\draw (5,0) node[and] (r) {\rotatebox{270}{}};
\draw (4.5,-1) node[or] (f) {\rotatebox{270}{}};
\draw (5,-2) node[bas] (b) {$a_2$};
\draw (4,-2) node[bas] (a) {$a_1$};
\draw (a) -- (f) -- (r) -- (b);
\draw (f) -- (b);
\draw (4.5,-3) node {$\tau_{\sigma}T$};
\end{tikzpicture}
\caption{An example of $\tau_{\sigma}T$ for surjective $\sigma$ and AT $T$.} \label{fig:scop}
\end{figure}
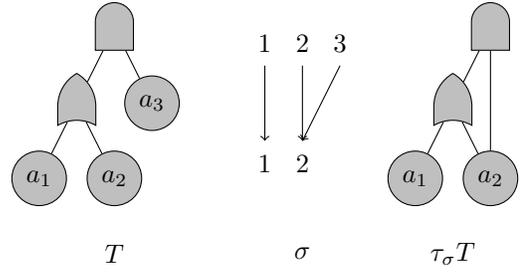

\begin{example} \label{exa:scoperad}
\begin{enumerate}
\item For an object $X$ in a category with finite products, consider the operad $\OEnd(X)$. We can consider this as a scoperad, by defining the map $\tau_{\sigma}\colon \End_n(X) \rightarrow \End_{n'}(X)$ as follows: For $i' \in [n']$, let $\varrho_{i'}\colon X \rightarrow X^{\sigma^{-1}(i')}$ be the diagonal embedding, and let $\varrho = \prod_{i' \in I'} \varrho_{i'}\colon X^{n'} \rightarrow X^n$; then for $f\colon X^n \rightarrow X$ one has $\tau_{\sigma}(f) = f \circ \varrho$. If $X$ is a set, then for $\vec{x} \in X^n$ we have $\tau_{\sigma}(f)(\vec{x}) = f(\sigma^*\vec{x})$, where $(\sigma^*\vec{x})_i = x_{\sigma(i)}$ for all $i$. This is the same as for bijective $f$, except now a coefficient $x_{\sigma(i)}$ can appear multiple times.
%\item The operad $\underline{\Delta}$ can be extended to a scoperad by taking $\tau_{\sigma}(p_1,\ldots,p_n) = \left(\sum_{i \in \sigma^{-1}(i')} p_i\right)_{i' \leq n'} \in \Delta_{n'}$. This comes down to the following: if $p$ is the probability distribution of a random variable $R \in [n]$, then $\tau_{\sigma}(p)$ is the distribution of $\sigma(R) \in [n']$.
\item Let $T \in \AT_n$ and let $\sigma\colon [n] \twoheadrightarrow [n']$ be surjective. Define $\tau_{\sigma}(T)$ by merging two BASs $a_i,a_{i'}$ whenever $\sigma(i) = \sigma(i')$; see Fig.~\ref{fig:scop} for an example. Thus $\tau_{\sigma}(T) = (N',E',\gamma')$ where
\begin{align*}
N' &= N \setminus B_T \cup \{\tilde{a}_1,\ldots,\tilde{a}_{n'}\},\\
E' &= \ldb (v,w) \in E \mid v,w \notin B_T \rdb\\
& \quad \quad \quad \uplus \ldb(v,\tilde{a}_{\sigma(i)}) \mid (v,a_i) \in E\rdb,\\
\gamma'(v) &= \begin{cases}
\gamma(v), & \text{ if $v \in N \setminus B_T$},\\
\BAS, & \text{ if $v = \tilde{a}_i$}.
\end{cases}
\end{align*}
This gives $\OAT$ the structure of a scoperad.
\end{enumerate}
\end{example}

%\begin{remark}
%One could also extend the definition of scoperads to have a map $\tau_{\sigma}$ for all $\sigma\colon[n] \rightarrow [n']$, not just surjective maps. Under that definition, $\OAT$ would not be a scoperad: The $\tau_{\sigma}$ corresponding to injective $\sigma$ would amount to adding BAS that are not connected to the root node, which is not allowed for ATs.
%\end{remark}

\subsection{The AT scoperad}

For us, the key advantage of defining scoperads is that now every AT can be decomposed into $T_{\tOR,n},T_{\tAND,n}$, as is shown by the following result:

\begin{definition}
An AT is called a \emph{prime AT} if it is isomorphic to $T_{\tOR,n}$ or $T_{\tAND,n}$ for some $n$.
\end{definition}

\begin{lemma} \label{lem:prime}
In the scoperad $\OAT$, every AT not equal to $\id \in \AT_1$ can be obtained from prime ATs and operations of the form $\star$ and $\tau_{\sigma}$.
\end{lemma}

Based on Example \ref{exa:scoperad}, we can define a \emph{scoperad AT metric} to be a scoperad morphism $\varphi\colon \OAT \rightarrow \OEnd(X)$, i.e., an operad morphism that preserves $\tau_{\sigma}$ for general surjective $\sigma$. In the following theorem, we show that scoperad metrics are precisely those that can be calculate by Algorithm \ref{alg:BU}. This gives a necessary and sufficient condition for when bottom-up methods work.

\begin{theorem} \label{thm:BUalg}
Let $\varphi\colon \OAT \rightarrow \OEnd(X)$ be an AT-metric. Then $\BU(T,\varphi,\vec{x},\R{T}) = \varphi(T)(\vec{x})$ for all $T,\vec{x}$ if and only if $\varphi$ is a scoperad morphism.
\end{theorem}

For the propositional metrics of Example \ref{exa:metric}.2, Algorithm \ref{alg:BU} has been shown not to work for DAG-like ATs in general, because BAS with multiple parents will occur twice somewhere in the calculation. Theorem \ref{thm:BUalg} shows that this is the only obstruction: if two BAS can be merged or split without consequences for the metric, then the bottom-up approach works.

\begin{remark}
Bottom-up algorithms have also been proposed to calculate metrics on extensions of ATs, by defining new operators for the new gates that have been introduced \cite{kordy2018quantitative,9925106}. The results have been more or less analogous: bottom-up algorithms succeed at calculating metrics on treelike ATs, but only work on general ATs under additional assumptions, such as idempotence of the operators.
\end{remark}

\section{BDD-based algorithms}

Binary decision diagrams (BDDs) form a compact way of representing Boolean functions. Applying these to the structure function of ATs, one can obtain efficient methods to calculate metrics, provided that the metric only depends on the structure function \cite{rauzy1993new,9925106}. In this section, we unify these existing methods under a common, generic algorithm for our operad algebra metrics, and we give sufficient conditions when this algorithm works. For the metrics to which it applies it is the state-of-the-art \cite{9925106}, again showing that there is no loss of efficiency in adopting the generic operad algebra framework. We again assume that all metrics $\varphi\colon \OAT \rightarrow \OEnd(X)$ are in the category $\cat{Set}$.

\subsection{Reduced ordered BDDs}

A BDD is a rooted DAG with two leaves $\tzero$ and $\tone$, and each nonleaf is labeled with a variable $F_1,\ldots,F_n$. Furthermore, each nonleaf has exactly two outgoing edges, labeled $\tzero$ and $\tone$. A BDD encodes a function $f\colon \BB^n \rightarrow \BB$ as follows: given $\vec{b} \in \BB^n$, starting from the root, at each node $v$ with label $F_i$, follow its $\tzero$- or $\tone$-edge depending on the value of $b_{i}$. The value of $f(\vec{b})$ is equal to the leaf we end up in. This allows us to represent Boolean functions in a way that is typically compact \cite{rauzy1997exact,bobbio2013methodology,}, making it suitable for storage and computation.

In this paper, we used \emph{reduced ordered BDDs}, which are BDDs in which the variables $F_i$ occur in order, which are then reduced to be of minimal size by getting rid of redundant and duplicate nodes. The formal definition is as follows:

\begin{definition}
Let $n \geq 0$. A \emph{reduced ordered binary decision diagram (ROBDD)} of $n$ variables is a tuple $B = (N,E,t_N,t_E)$ where:
\begin{itemize}
\item $(N,E)$ is a rooted DAG where every nonleaf has exactly two children.
\item $t_N\colon N \rightarrow \{F_1,\ldots,F_n\} \cup \{\tzero,\tone\}$ satisfies $t_N(v) \in \{\tzero,\tone\}$ iff $v$ is a leaf.
\item $t_E\colon E \rightarrow \{\tzero,\tone\}$ is such that the two outgoing edges of each nonleaf $v$ have different values; the two children are denoted $c_{\tzero}(v)$ and $c_{\tone}(v)$ depending on the edge values.
\item If $(v,v')$ is an edge and $t_N(v) = F_i$, $t_N(v') = F_{i'}$, then $i < i'$.
\item If $v,v'$ are two nodes such that $t_N(v) = t_N(v')$, and for nonleaves furthermore $c_{\tzero}(v) = c_{\tzero}(v')$ and $c_{\tone}(v) = c_{\tone}(v')$, then $v = v'$.
\end{itemize}
\end{definition}

\begin{figure}
\centering
\begin{tikzpicture}[
and/.style={and gate US,rotate=90,draw,fill = lightgray},
or/.style={or gate US,rotate=90,draw,fill = lightgray},
bas/.style={circle,draw,fill = lightgray},
blok/.style={rectangle,draw,fill = lightgray}]
\draw (0,-2) -- (0.5,-1);
\draw (0,0) node[and] (r) {\rotatebox{270}{}};
\draw (-0.5,-1) node[or] (f) {\rotatebox{270}{}};
\draw (0.5,-1) node[or] (s) {\rotatebox{270}{}};
\draw (0,-2) node[bas] (b) {$a_1$};
\draw (1,-2) node[bas] (c) {$a_3$};
\draw (-1,-2) node[bas] (a) {$a_2$};
\draw (f) -- (r) -- (s);
\draw (s) -- (c);
\draw (a) -- (f) -- (b);
\draw (0,-2.7) node {$0.7$};
\draw (1,-2.7) node {$0.3$};
\draw (-1,-2.7) node {$0.5$};

\draw (3.5,0) node[blok] (t1) {$F_1$};
\draw (4.5,-1) node[blok] (t2) {$F_2$};
\draw (3.5,-2) node[blok] (t3) {$F_3$};
\draw (2.5,-3) node[blok] (t4) {$\tzero$};
\draw (4.5,-3) node[blok] (t5) {$\tone$};
\path (t1) edge node[above right] {$\tone$} (t2)
(t2) edge node[right] {$\tone$} (t5)
(t3) edge node[above] {$\tone$} (t5);
\path[dashed] (t1) edge node[left] {$\tzero$} (t4)
(t2) edge node[above left] {$\tzero$} (t3)
(t3) edge node[above] {$\tzero$} (t4);
\end{tikzpicture}
\caption{An AT (with probability values) and its ROBDD.} \label{fig:bdd}
\end{figure}
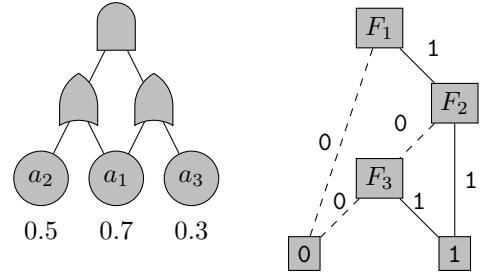

Each Boolean function can be represented by a unique ROBDD \cite{bryant1992symbolic}. The size of the ROBDD representing $f$ is worst-case exponential in $n$, but in practice it is often quite compact \cite{bryant1986graph}. The size depends heavily on the ordering of the $n$ variables, but finding the optimal ordering is NP-complete.

\subsection{The BDD algorithm}

The BDD of $\struc{T}$ has been used to efficiently calculate AT metrics, see below. One necessary condition is that an AT's metric value only depends on its structure function. This is formalized in the definition below. Recall that $\OMBool = \OEnd_{\cat{Pos}}(\BB)$ is the operad of monotonous Boolean functions.

\begin{definition}
Let $\OMBool^{\recht{nc}} \subset \OMBool$ be the sub-operad of nonconstant Boolean functions (so $\MBool^{\recht{nc}}_0 = \varnothing$). An AT metric $\varphi\colon \OAT \rightarrow \OEnd(X)$ is called \emph{propositional} if there exists a $\psi\colon\OMBool^{\recht{nc}} \rightarrow \OEnd(X)$ such that $\varphi = \psi \circ \varphi^{\recht{struc}}$, where $\varphi^{\recht{struc}}\colon \OAT \rightarrow \OMBool^{\recht{nc}}$ is the structure function operad of Example \ref{exa:metric}.1.
\end{definition}

In Example \ref{exa:metric}.1, the map $\varphi^{\recht{struc}}$ has codomain $\OMBool$, but structure functions are always nonconstant: one has $\struc{T}((\tzero,\ldots,\tzero),\R{T}) = \tzero$ and $\struc{T}((\tone,\ldots,\tone),\R{T}) = \tone$. Therefore, one can take $\OMBool^{\recht{nc}}$ for the codomain. In fact, if one does so, then each $\varphi^{\recht{struc}}\colon \AT_n \rightarrow \MBool^{\recht{nc}}_n$ is surjective; hence $\psi$, if it exists, is uniquely defined.

Some examples of propositional metrics are $\varphi^{\recht{struc}}$ itself (trivial) and the metrics $\varphi^{\recht{TAP}}$ and $\varphi^{\recht{ps}}_D$ from Example \ref{exa:metric}. A non-example is given by the bottom-up semiring metrics of Example \ref{exa:metric}.4).

For propositional metrics, one can formulate metric-calculating algorithms by traversing the ROBDD bottom-up. This is presented in Algorithm \ref{alg:BDD}. Such an algorithm requires, apart from $X$, a function $g\colon X^3 \rightarrow X$ prescribing how a node's value should be computed, from the values of its two children and the BAS value of its label. We also need initial values $z_{\tzero},z_{\tone} \in X$. Of course, $g,z_{\tzero},z_{\tone}$ should be chosen so that the algorithm calculates $\varphi(T)(\vec{x})$ as intended.

Algorithm \ref{alg:BDD} has only linear complexity in the size of $B$, and as such is very fast when $B$ is small. Unfortunately, $B$ is of worst-case exponential size. Nevertheless, in practice $B$ is often quite small, making BDD-based algorithms the state-of-the-art approach for the two metrics described in the example below \cite{9925106}, when these metrics cannot be covered by the bottom-up algorithm of Section \ref{sec:bu}.

\begin{algorithm}[t]
	\KwIn{Set $X$, function $g\colon X^3 \rightarrow X$, initial values $z_{\tzero},z_{\tone} \in X$; AT $T = (N,E,\gamma) \in \AT_n$, BAS values $\vec{x} \in X^n$}
	\KwOut{$\mathtt{BDD}(X,g,z_{\tzero},z_{\tone},T,\vec{x}) = \varphi(T)(\vec{x}) \in X$}
	\BlankLine
        $B \leftarrow$ ROBDD representing $\struc{T}(\bullet,\R{T})$\;
        \Return{$\mathtt{BDD\_{BU}}(B,x,\R{B})$}
\vspace{0.25em}        
	\caption{The algorithm $\mathtt{BDD}$ (see Algorithm \ref{alg:BDDBU} for $\mathtt{BDD\_BU}$).}
	\label{alg:BDD}
\end{algorithm}

\begin{algorithm}[t]
	\KwIn{Set $X$, function $g\colon X^3 \rightarrow X$, initial values $z_{\tzero},z_{\tone} \in X$; ROBDD $B = (N,E,t_N,t_E)$, BAS values $\vec{x} \in X^n$,node $v \in N$}
	\KwOut{value $\mathtt{BDD\_BU}(B,\vec{x},v) \in X$}
	\BlankLine
 \uIf{$t_N(v) = b_i$}{
 \Return{$g(x_i,\mathtt{BDD\_BU}(B,\vec{x},c_{\tzero}(v)),\mathtt{BDD\_BU}(B,\vec{x},c_{\tone}(v))$}} 
        \uElseIf{$t_N(v) = \tzero$}{%
		\Return{$z_{\tzero}$}
	} \Else{
\Return{$z_{\tone}$}
	} 
	\caption{The algorithm $\mathtt{BDD\_BU}$. For notational convenience, the arguments $X,g,z_{\tzero},z_{\tone}$ are omitted from the notation.}
	\label{alg:BDDBU}
\end{algorithm}

\begin{example} \label{exa:bdd} We consider two of the metrics from Example \ref{exa:metric}.
\begin{enumerate}
\item The metric $\varphi^{\recht{TAP}}$ of Example \ref{exa:metric}.4 can be calculated this way, by taking $z_{\tzero} = 0$, $z_{\tone} = 1$, and $g(p,q,r) = (1-p)q+pr$ \cite{rauzy1993new}. For example, consider the AT $T$ of Fig.~\ref{fig:bdd} and its ROBDD $B$; for the BAS values we take $\vec{x} = (0.7,0.5,0.3)$. Then we find $\varphi^{\recht{TAP}}(T)(\vec{x})$ by traversing the ROBDD bottom-up:
\begin{align*}
\mathtt{BDD\_BU}(B,\vec{x},\tzero) &= z_{\tzero} &&= 0,\\
\mathtt{BDD\_BU}(B,\vec{x},\tone) &= z_{\tone} &&= 1,\\
\mathtt{BDD\_BU}(B,\vec{x},F_3) &= (1-0.3) \cdot 0 + 0.3 \cdot 1 &&= 0.3,\\
\mathtt{BDD\_BU}(B,\vec{x},F_2) &= (1-0.5) \cdot 0.3 + 0.5 \cdot 1 &&= 0.65,\\
\mathtt{BDD\_BU}(B,\vec{x},F_2) &= (1-0.7) \cdot 0 + 0.7 \cdot 0.65 &&= 0.455.
\end{align*}
Hence $\varphi^{\recht{TAP}}(T)(\vec{x}) = 0.455$. Choosing a different variable order would have resulted in a different ROBDD, but the resulting probability would have been the same.
\item Consider the \emph{min cost} metric of Example \ref{exa:metric}.2; this is calculated by Algorithm \ref{alg:BDD}, for $z_{\tzero} = \infty$, $z_{\tone} = 0$, and $g(x,y,z) = \min(y,x+z)$. This construction can be applied to general propositional semiring metrics $D = (V,\triangle,\triangledown)$, provided that there exist identity elements $1_{\triangle}$,$1_{\triangledown}$ for $\triangle$ and $\triangledown$, and that the semiring is \emph{absorbing}, i.e., $x \triangledown (x \triangle y) = x$ for all $x,y \in V$. Then the propositional semiring metric $\varphi^{\recht{ps}}_D$ can be found using $z_{\tzero} = 1_{\triangledown}$, $z_{\tone} = 1_{\triangle}$, and $g(x,y,z) = y \triangledown (x \triangle z)$. For more details see \cite{9925106}.
% \footnote{\cite{9925106} does not assume that $D$ is absorbing, but the result is not true in the generality claimed there: counterexamples can be constructed in the semiring $(\mathbb{R}_{\geq 0},+,\max)$.}
\end{enumerate}
\end{example}

\subsection{BDD algorithms for general metrics}

In this section, we give sufficient conditions for the validity of Algorithm \ref{alg:BDD}. The main conditions are that a metric $\varphi$ preserves certain constructions, which are given in Definitions \ref{def:shannon} and \ref{def:expansion} below.

For each $n$, define a partial order $\preceq$ on $\MBool_n$ by $f \preceq f'$ if $f(\vec{b}) \leq f'(\vec{b})$ for all $\vec{b} \in \BB^n$. We use this to define \emph{Shannon composition}, which for monotonous functions is the reverse operation of the well-known Shannon expansion: that is, $f$ and $f'$ are the Shannon cofactors of $\recht{Sh}(f,f')$. 

\begin{definition}\label{def:shannon}
Let $f,f' \in \MBool_{n-1}$ with $f \prec f'$. The \emph{Shannon composition} of $f,f'$ is the map $\recht{Sh}(f,f') \in \MBool_n$ defined by
\begin{equation*}
\recht{Sh}(f,f')(x_1,\ldots,x_n) = f(x_2,\ldots,x_n) \vee (x_1 \wedge f'(x_2,\ldots,x_n)).
\end{equation*}
\end{definition}

The second operation we want to introduce is a formality: we create a $(n+1)$-ary function from an $n$-ary function by adding an irrelevant argument.

\begin{definition} \label{def:expansion}
For $X$, and $n \in \mathbb{Z}_{\geq 0}$, define $\delta_n\colon \End_n(X) \rightarrow \End_{n+1}(X)$ by $(\delta_n f)(x_1,\ldots,x_{n+1}) = f(x_2,\ldots,x_{n+1})$ for $x \in X^{n+1}$.
\end{definition}

Since $\OMBool = \OEnd_{\cat{Pos}}(\mathbb{B})$, this also defines $\delta_n$ on $\MBool_n$.

\begin{wrapfigure}[6]{r}{3cm}
\centering
\begin{tikzpicture}[
and/.style={rectangle,draw,fill = lightgray},
bas/.style={circle,draw,fill = lightgray}]
\draw (0,0) node[and] (r) {$F_1$};
\draw (-0.5,-1) node[and] (f) {$B_f$};
\draw (0.5,-1) node[and] (s) {$B_{f'}$};
\draw[dashed] (f) --node[left]{$\tzero$} (r);
\draw (r) -- node[right]{$\tone$} (s);
\end{tikzpicture}
\end{wrapfigure}
Every function in $\OMBool$ can be created from the constant functions $\tzero,\tone \in \MBool_0$ through repeated applications of $\recht{Sh}$ and $\delta_n$. Furthermore, on BDDs these operations have a straightforward representation: $f$ and $\delta_nf$ are represented by the same ROBDD, and $\recht{Sh}(f,f')$ is represented by the BDD on the right, where $B_f$ and $B_{f'}$ are ROBDDs representing $f$ and $f'$; this BDD can be reduced to a ROBDD by identifying identical nodes on the left and right. Theorem \ref{thm:bdd} uses this idea to connect BDDs to metric calculation: if $\delta_n$ is irrelevant to a metric $\varphi$, and if $\recht{Sh}$ is mapped to the application of $g$, then Algorithm \ref{alg:BDD} correctly calculates $\varphi$.

\begin{theorem} \label{thm:bdd}
Let $\varphi\colon \OAT \rightarrow \OEnd(X)$ be propositional, and let $\psi\colon \OMBool^{\recht{nc}} \rightarrow \OEnd(X)$ be such that $\varphi = \psi \circ \phi^{\recht{struc}}$. Let $\iota$ be the inclusion $\OMBool^{\recht{nc}} \hookrightarrow \OMBool$. Let $g\colon X^3 \rightarrow X$ and $z_{\tzero},z_{\tone} \in X$. Suppose there exists a $\Psi\colon \OMBool \rightarrow \OEnd(X)$ satisfying the following conditions:
\begin{enumerate}
\item $\psi = \Psi \circ \iota$;
\item $\Psi(\tzero) = z_{\tzero}$ and $\Psi(\tone) = z_{\tone}$;
\item $\Psi \circ \delta_n = \delta_n \circ \Psi$ for all $n$;
\item For all $n \geq 1$ and $f,f' \in \MBool_{n-1}$ with $f \prec f'$, and all $\vec{x} \in X^n$, one has
\begin{align*}
&\Psi(\recht{Sh}(f,f'))(\vec{x})\\
&= g\left(x_1,\Psi(f)(x_2,\ldots,x_n),\Psi(f')(x_2,\ldots,x_n)\right).
\end{align*}
\end{enumerate}
Then $\mathtt{BDD}(X,g,z_{\tzero},z_{\tone},T,x) = \varphi(T)(\vec{x})$ for all $T \in \AT_n$ and $\vec{x} \in X^n$.
\end{theorem}

Theorem \ref{thm:bdd} requires a number of conditions, but they are not all equally hard to satisfy. The requirements that $\varphi$ be propositional and that $\delta$ is preserved are necessary conditions: this is because the ROBDD depends only on the structure function, and because adding irrelevant variables does not change the ROBDD. Furthermore, the extension $\Psi$, if it exists, is unique: the constant functions $\tzero$ and $\tone$ have to be sent to the neutral elements for the operators $C_2$ and $D_2$, respectively. This also shows that $z_{\tzero}$ and $z_{\tone}$ are determined by $\psi$. The key property needed to make Algorithm \ref{alg:BDD} work is the fact that $\Psi$ maps Shannon composition to $g$.

\begin{example}
Consider $\varphi^{\recht{TAP}}$ of Example \ref{exa:metric}.4. For a given $\vec{p} \in [0,1]^n$, let $Y_i \sim \recht{Ber}(p_i)$, with all $Y_i$ independent; then $\varphi^{\recht{TAP}}(T)(\vec{p}) = \Pr[\struc{T}(\vec{Y},\R{T}) = \tone]$. We can extend $\varphi^{\recht{TAP}}$ to a $\Psi\colon \OMBool \rightarrow \OEnd([0,1])$ by $\Psi(f)(\vec{p})= \Pr[f(Y) = \tone]$. We show that this maps Shannon composition to the $g$ of Example \ref{exa:bdd}.1). For $f \prec f'$ one has $f'(Y_2,\ldots,Y_n) = \tone$ whenever $f(Y_2,\ldots,Y_n) = \tone$; hence $\recht{Sh}(f,f')(Y) = \tone$ if either $Y_1 = \tone$ and $f'(Y_2,\ldots,Y_n) = \tone$, or $Y_1 = \tzero$ and $f(Y_2,\ldots,Y_n) = \tone$. Hence
\begin{align*}
\Psi(\recht{Sh}(f,f'))(\vec{p}) &= (1-p_1)\Pr[f(Y_2,\ldots,Y_n)=\tone]\\
& \quad \quad +p_1\Pr[f'(Y_2,\ldots,Y_n)=\tone] \\
&= g(p_1,\Psi(f)(p_2,\ldots,p_n),\Psi(f')(p_2,\ldots,p_n)).
\end{align*}
This shows condition 4 of Theorem \ref{thm:bdd}. The other conditions are straightforward to check, and together show that Algorithm \ref{alg:BDD} works for $\varphi^{\recht{TAP}}$.

\end{example}

\section{Conclusion}

This work introduces a new definition of AT metrics based on operad theory. This definition captures all existing definitions, and can be used in a wider, category-theoretical context. It also generalizes existing metric calculation algorithms and gave operad-theoretical conditions as to when they apply. Future work can consider metric calculation on extensions such as dynamic ATs or attack-defense trees, or investigate to what extent our work applies to fault trees.

\section*{Acknowledgements}

This research has been partially funded by ERC Consolidator grant 864075 CAESAR and the European Union’s Horizon 2020 research and innovation programme under the Marie Skłodowska-Curie grant agreement No. 101008233.

\printbibliography

@book{limnios2013fault,
  title={Fault trees},
  author={Limnios, Nikolaos},
  year={2013},
  publisher={John Wiley \& Sons}
}

@book{yau2016colored,
  title={Colored operads},
  author={Yau, Donald},
  volume={170},
  year={2016},
  publisher={American Mathematical Society}
}

@inproceedings{beckers2014determining,
  title={Determining the probability of smart grid attacks by combining attack tree and attack graph analysis},
  author={Beckers, Kristian and Heisel, Maritta and Krautsevich, Leanid and Martinelli, Fabio and Meis, Rene and Yautsiukhin, Artsiom},
  booktitle={International Workshop on Smart Grid Security},
  pages={30--47},
  year={2014},
  organization={Springer}
}

@book{yau2018operads,
  title={Operads of wiring diagrams},
  author={Yau, Donald},
  volume={2192},
  year={2018},
  publisher={Springer}
}

@ARTICLE{9925106,
  author={Lopuhaä-Zwakenberg, Milan and Budde, Carlos E. and Stoelinga, Mariëlle},
  journal={IEEE Transactions on Dependable and Secure Computing}, 
  title={Efficient and Generic Algorithms for Quantitative Attack Tree Analysis}, 
  year={2022},
  volume={},
  number={},
  pages={1-18},
  doi={10.1109/TDSC.2022.3215752}}

@article{baez2015operads,
  title={Operads and phylogenetic trees},
  author={Baez, John C and Otter, Nina},
  journal={arXiv:1512.03337},
  year={2015},
  note={Preprint}
}

@incollection{meseguer1989general,
  title={General logics},
  author={Meseguer, Jos{\'e}},
  booktitle={Studies in Logic and the Foundations of Mathematics},
  volume={129},
  pages={275--329},
  year={1989},
  publisher={Elsevier}
}

@article{doubek2017properads,
  title={Properads and Homotopy Algebras Related to Surfaces},
  author={Doubek, Martin and Jurco, Branislav and Peksova, Lada},
  journal={arXiv preprint arXiv:1708.01195},
  year={2017}
}

@article{foley2021operads,
  title={Operads for complex system design specification, analysis and synthesis},
  author={Foley, John D and Breiner, Spencer and Subrahmanian, Eswaran and Dusel, John M},
  journal={Proceedings of the Royal Society A},
  volume={477},
  number={2250},
  pages={20210099},
  year={2021},
  publisher={The Royal Society Publishing}
}

@article{bradley2021entropy,
  title={Entropy as a topological operad derivation},
  author={Bradley, Tai-Danae},
  journal={Entropy},
  volume={23},
  number={9},
  pages={1195},
  year={2021},
  publisher={MDPI}
}

@inproceedings{kordy2018quantitative,
  title={On quantitative analysis of attack--defense trees with repeated labels},
  author={Kordy, Barbara and Wide{\l}, Wojciech},
  booktitle={International Conference on Principles of Security and Trust},
  pages={325--346},
  year={2018},
  organization={Springer}
}

@article{schneier1999attack,
  title={Attack trees},
  author={Schneier, Bruce},
  journal={Dr. Dobb’s journal},
  volume={24},
  number={12},
  pages={21--29},
  year={1999}
}

@inproceedings{fila2019efficient,
  title={Efficient attack-defense tree analysis using Pareto attribute domains},
  author={Fila, Barbara and Wide{\l}, Wojciech},
  booktitle={2019 IEEE 32nd Computer Security Foundations Symposium (CSF)},
  pages={200--20015},
  year={2019},
  organization={IEEE}
}

@inproceedings{khand2007attack,
  title={An attack model development process for the cyber security of safety related nuclear digital I\&C systems},
  author={Khand, Parvaiz Ahmed and Seong, Poong Hyun},
  booktitle={Proceedings of the Korean Nucleary Society (KNS) Fall meeting},
  year={2007}
}

@inproceedings{bossuat2017evil,
  title={Evil twins: handling repetitions in attack--defense trees},
  author={Bossuat, Ang{\`e}le and Kordy, Barbara},
  booktitle={International Workshop on Graphical Models for Security},
  pages={17--37},
  year={2017},
  organization={Springer}
}

@inproceedings{dong2017attack,
  title={An attack tree-based approach for vulnerability assessment of communication-based train control systems},
  author={Dong, Huiyu and Wang, Hongwei and Tang, Tao},
  booktitle={2017 Chinese Automation Congress (CAC)},
  pages={6407--6412},
  year={2017},
  organization={IEEE}
}

@inproceedings{kumar2015quantitative,
  title={Quantitative attack tree analysis via priced timed automata},
  author={Kumar, Rajesh and Ruijters, Enno and Stoelinga, Mari{\"e}lle},
  booktitle={International Conference on Formal Modeling and Analysis of Timed Systems},
  pages={156--171},
  year={2015},
  organization={Springer}
}

@inproceedings{kordy2010foundations,
  title={Foundations of attack--defense trees},
  author={Kordy, Barbara and Mauw, Sjouke and Radomirovi{\'c}, Sa{\v{s}}a and Schweitzer, Patrick},
  booktitle={International Workshop on Formal Aspects in Security and Trust},
  pages={80--95},
  year={2010},
  organization={Springer}
}

@incollection{mcqueen2006time,
  title={Time-to-compromise model for cyber risk reduction estimation},
  author={McQueen, Miles A. and Boyer, Wayne F. and Flynn, Mark A. and Beitel, George A.},
  booktitle={Quality of protection},
  pages={49--64},
  year={2006},
  publisher={Springer}
}

@article{bryant1992symbolic,
  title={Symbolic boolean manipulation with ordered binary-decision diagrams},
  author={Bryant, Randal E.},
  journal={ACM Computing Surveys (CSUR)},
  volume={24},
  number={3},
  pages={293--318},
  year={1992},
  publisher={ACM New York, NY, USA}
}

@article{bryant1986graph,
  title={Graph-based algorithms for boolean function manipulation},
  author={Bryant, Randal E.},
  journal={Computers, IEEE Transactions on},
  volume={100},
  number={8},
  pages={677--691},
  year={1986},
  publisher={IEEE}
}

@article{bobbio2013methodology,
  title={A methodology for qualitative/quantitative analysis of weighted attack trees},
  author={Bobbio, Andrea and Egidi, Lavinia and Terruggia, Roberta},
  journal={IFAC Proceedings Volumes},
  volume={46},
  number={22},
  pages={133--138},
  year={2013},
  publisher={Elsevier}
}

@article{rauzy1997exact,
  title={Exact and truncated computations of prime implicants of coherent and non-coherent fault trees within Aralia},
  author={Rauzy, Antoine and Dutuit, Yves},
  journal={Reliability Engineering \& System Safety},
  volume={58},
  number={2},
  pages={127--144},
  year={1997},
  publisher={Elsevier}
}

@article{lopuhaa2021attack,
  title={Attack time analysis in dynamic attack trees via integer linear programming},
  author={Lopuha{\"a}-Zwakenberg, Milan and Stoelinga, Mari{\"e}lle},
  journal={arXiv preprint arXiv:2111.05114},
  year={2021}
}

@article{rauzy1993new,
  title={New algorithms for fault trees analysis},
  author={Rauzy, Antoine},
  journal={Reliability Engineering \& System Safety},
  volume={40},
  number={3},
  pages={203--211},
  year={1993},
  publisher={Elsevier}
}

@inproceedings{mauw2005foundations,
  title={Foundations of attack trees},
  author={Mauw, Sjouke and Oostdijk, Martijn},
  booktitle={International Conference on Information Security and Cryptology},
  pages={186--198},
  year={2005},
  organization={Springer}
}

@inproceedings{jhawar2015attack,
  title={Attack trees with sequential conjunction},
  author={Jhawar, Ravi and Kordy, Barbara and Mauw, Sjouke and Radomirovi{\'c}, Sa{\v{s}}a and Trujillo-Rasua, Rolando},
  booktitle={IFIP International Information Security and Privacy Conference},
  pages={339--353},
  year={2015},
  organization={Springer}
}

@book{markl2002operads,
  title={Operads in algebra, topology and physics},
  author={Markl, Martin and Shnider, Steven and Stasheff, James D},
  volume={96},
  year={2002},
  publisher={American Mathematical Society Providence, RI}
}

@inproceedings{budde2021efficient,
  title={Efficient algorithms for quantitative attack tree analysis},
  author={Budde, Carlos E and Stoelinga, Mari{\"e}lle},
  booktitle={2021 IEEE 34th Computer Security Foundations Symposium (CSF)},
  pages={1--15},
  year={2021},
  organization={IEEE}
}

@article{dutuit1996linear,
  title={A linear-time algorithm to find modules of fault trees},
  author={Dutuit, Yves and Rauzy, Antoine},
  journal={IEEE Transactions on Reliability},
  volume={45},
  number={3},
  pages={422--425},
  year={1996},
  publisher={IEEE}
}

\appendices

\section{Proof of Theorem \ref{thm:BU}}

\begin{proof}
We prove this by bottom-up induction on $v$. If $v$ is the leaf $a_i$, then $\gamma(v) = \BAS$, and $T_v = \id \in \AT_1$. Since $\varphi$ is an operad morphism, $\varphi(T_v) = \id$, and $\varphi(T_v)(\vec{x}_v) = x_i = \BU(T,\varphi,\vec{x},v)$.

Now suppose $v = \tOR(w_1,\ldots,w_m)$, and the statement is true for $w_1,\ldots,w_m$. The recursive definition of $\BU(T,\varphi,\vec{x},w_i)$ depends only on the subDAG $T_{w_i}$. As such one can prove by a straightforward induction proof that
\begin{equation} \label{eq:pfbu1}
\BU(T,\varphi,\vec{x},w_i) = \BU(T_{w_i},\varphi,\vec{x}_{w_i},w_i).
\end{equation}
Now, since $T$ is a tree, the subtrees $T_{w_1},\ldots,T_{w_m}$ are disjoint, and 
\begin{align*}
T &= T_{\tOR,m}[T_{w_1}/a_1,\ldots,T_{w_m}/a_m] \\
&= T_{\tOR,m} \star (T_{w_1},\ldots,T_{w_m}).
\end{align*}
Since $\varphi$ is an operad morphism, we get
\begin{align}
&\varphi(T)(\vec{x}) \nonumber\\
&= \varphi(T_{\tOR,m} \star (T_{w_1},\ldots,T_{w_m}))(\vec{x})\nonumber \\
&= \varphi(T_{\tOR,m})(\varphi(T_{w_1})(\vec{x}_{w_1}),\ldots,\varphi(T_{w_1})(\vec{x}_{w_m})) \label{eq:pfbu2}\\
&= C_m(\BU(T_{w_1},\varphi,\vec{x}_{w_1},w_1),\ldots,\BU(T_{w_m},\varphi,\vec{x}_{w_m},w_m)) \nonumber\\
& = C_m(\BU(T,\varphi,\vec{x},w_1),\ldots,\BU(T,\varphi,\vec{x},w_m)) \label{eq:pfbu3}\\
&= \BU(T,\varphi,\vec{x},v).\nonumber
\end{align}
Here we use the induction hypothesis in \eqref{eq:pfbu3} and \eqref{eq:pfbu1} in \eqref{eq:pfbu2}. This proves the theorem for $\gamma(v) = \tOR$; the case that $\gamma(v) = \tAND$ is completely analogous.
% We prove this by induction on the size of $T$. If $\gamma(\R{T}) = \BAS$, then $T = \id \in \AT_1$; since $\varphi$ is an operad morphism $T$ is then mapped to the identity on $X$, so $\BU(T,\varphi,x) = x = \varphi(T)(x)$.
% Suppose $T \neq \id$, and let $v$ be a non-BAS with only BAS children. Suppose $\gamma(v) = \tOR$, and WLOG assume that the children of $v$ are $a_1,\ldots,a_k$. Let $T'$ be the AT obtained by contracting $T_v$ into a BAS $v$; then $T = T'[T_{\tOR,k}/v]$. By the induction hypothesis, $\BU$ correctly calculates $\varphi$ for $T'$. Then since $\varphi$ is an operad morphism, it respects $\star$ and we have
% \begin{align*}
% &\varphi(T)(\vec{x}) \\ &= \varphi(T'[T_{\tOR,k}/v])(\vec{x}) \\
% &= \varphi(T' \star (T_{\tOR,k},\underbrace{\id,\ldots,\id}_{n-k \times}))(\vec{x}) \\
% & = \varphi(T')(\varphi(T_{\tOR,k})(x_1,\ldots,x_k),\varphi(\id)(x_{k+1}),\ldots,\varphi(\id)(x_n)) \\
% &= \varphi(T')(C_k(x_1,\ldots,x_k),x_{k+1},\ldots,x_n) \\
% \end{align*}
% \begin{align*}
% \varphi(T)(\vec{x}) &= \varphi(T'[T_{\tOR,k}/v])(\vec{x}) \\
% &= \varphi(T')(C_k(x_1,\ldots,x_k),x_{k+1},\ldots,x_n) \\
% &= \BU(T',\varphi,(C_k(x_1,\ldots,x_k),x_{k+1},\ldots,x_n)) \\
% &= \BU(T,\varphi,x).
% \end{align*}
% The case that $\gamma(v) = \tAND$ is completely analogous.
\end{proof}

\section{Proof of Lemma \ref{lem:prime}}

\begin{proof}
Define a preorder $\preceq$ on the set of (isomorphism classes) of ATs by $T \preceq T'$ if $T$ has less non-BAS nodes that $T'$, or if $T$ has as much non-BAS nodes as $T'$, and $T$ has \emph{more} BAS nodes with multiple parents than $T'$. We prove the lemma by induction over $\preceq$. The basis for the induction are the $T$ with only 1 non-BAS node. These are not necessarily of the form $T_{\tOR,n}$ or $T_{\tAND,n}$, as there may be multiple edges between the root and a given BAS. However, we can obtain $T$ from a prime AT by merging BASs via a $\tau_{\sigma}$. More precisely, suppose $T = (N,E,\gamma)$; then $N = \{\R{T},a_1,\ldots,a_n\}$; let $E = \ldb e_1,\ldots,e_m \rdb$. Define $\sigma\colon [m] \twoheadrightarrow [n]$ by $\sigma(i) = j$ if and only if $e_i = (\R{T},a_j)$. Then $T = \tau_{\sigma} T'$, where $T'$ is either $T_{\tOR,m}$ or $T_{\tAND,m}$ depending on $\gamma(\R{T})$. This proves the induction basis.

Now let $T$ be an AT with more than 1 non-BAS node, and let $v$ be a node of $T$ with only BAS children. If every child of $v$ has only $v$ as a parent, then $T = T'[T_v/v]$, where $T' \prec T$ and $T_v$ has only 1 non-BAS node; hence by the induction hypothesis $T$ can be constructed from prime trees. If such a $v$ does not exist, there is a BAS $a$ with multiple parents. Construct a new tree $T'$ with multiple independent copies of $a$, such that each parent is only connected to a single copy. Then $T' \prec T$, and furthermore $T = \tau_{\sigma}(T')$ for a $\sigma$ that merges all the copies of $b$. By the induction hypothesis $T'$, and therefore $T$, can be constructed from prime trees.
\end{proof}

\section{Proof of Theorem \ref{thm:BUalg}}

For the proof we need one auxiliary lemma.

\begin{lemma} \label{lem:busigma}
Let $T \in \AT_n$ be an AT, let $\sigma\colon [n] \rightarrow [n']$ be a surjective map, and let $\tau_{\sigma}T = (N',E',\gamma')$ as in Example \ref{exa:scoperad}.2). Define a map $\pi\colon N \rightarrow N'$ by
\[
\pi(v) = \begin{cases}
a'_j, & \textrm{ if $v = a_i$ and $\sigma(i) = j$,}\\
v, & \textrm{ if $v \notin B_T$}.
\end{cases}
\]
Let $(\varphi,X)$ be an AT metric, and let $\vec{x} \in X^n$. Let $\sigma^*\vec{x}$ be as in Example \ref{exa:scoperad}.1). Then for all $v \in N$
\[
\BU(\tau_{\sigma}T,\varphi,\vec{x},\pi(v)) = \BU(T,\varphi,\sigma^*\vec{x},v).
\]
\end{lemma}

\begin{proof}
The proof is a straightforward induction on $T$. If $v$ is the BAS $a_i$ with $\sigma(i) = j$, then 
\begin{align*}
\BU(\tau_{\sigma}T,\varphi,\vec{x},\pi(v)) &= \BU(\tau_{\sigma}T,\varphi,\vec{x},a'_j) \\
&= x_j \\
&= (\sigma^*\vec{x})_i \\
&= \BU(T,\varphi,\sigma^*\vec{x},v).
\end{align*}
If $v = \tOR(w_1,\ldots,w_m)$, then $\pi(v) = \tOR(\pi(w_1),\ldots,\pi(w_m))$. From this we show the statement for $v$ given the induction hypothesis for the $w_i$, and the $\tAND$-case is analogous.
\end{proof}

\begin{proof}[Proof of Theorem \ref{thm:BUalg}]
First assume that $\varphi$ is a scoperad morphism. For all $n$ and all $T \in \AT_n$, we want to prove the property
\begin{equation}
\forall \vec{x} \in X^n \colon \BU(T,\varphi,\vec{x},\R{T}) = \varphi(T)(\vec{x}). \label{eq:toprove}
\end{equation}
Clearly \eqref{eq:toprove} holds for $\id$ and for prime ATs. We will show that \eqref{eq:toprove} is preserved under $\tau_{\sigma}$ and $\star$; by Lemma \ref{lem:prime} this shows that it holds for all $T$.

Suppose \eqref{eq:toprove} holds for $T \in \AT_n$, and let $\sigma\colon [n] \twoheadrightarrow [m]$ be a surjective map. Let $\vec{x} \in X^m$; we aim to show that $\BU(\tau_{\sigma}T,\varphi,\vec{x},\R{\tau_{\sigma}T}) = \varphi(\tau_{\sigma}T)(\vec{x})$. We then have
\begin{align}
\BU(\tau_{\sigma}T,\varphi,\vec{x},\R{\tau_{\sigma}T}) &= \BU(T,\varphi,\sigma^*\vec{x},\R{T}) \label{eq:bualg1}\\
&= \varphi(T)(\sigma^*{\vec{x}}) \label{eq:bualg2}\\
&= \tau_{\sigma}(\varphi(T))(\vec{x}) \label{eq:bualg3}\\
&= \varphi(\tau_{\sigma} T)(\vec{x}). \label{eq:bualg4}
\end{align}
Here \eqref{eq:bualg1} is due to Lemma \ref{lem:busigma}, \eqref{eq:bualg2} is because \eqref{eq:toprove} holds for $T$, \eqref{eq:bualg3} is from the definition of $\tau_{\sigma}$ on $\OEnd(X)$ from Example \ref{exa:scoperad}.1), and \eqref{eq:bualg4} is because $\varphi$ is a scoperad morphism. We conclude that \eqref{eq:toprove} holds for $\tau_{\sigma}(T)$.

Now assume \eqref{eq:toprove} holds for $T \in \AT_n$ and $T_1 \in \AT_{m_1},\ldots,T_n \in \AT_{m_n}$; we aim to prove it for $T'= T \star (T_1,\ldots,T_n)$. Hence, for every $(\vec{x}_{(i)})_{i \leq n} \in \prod_i X^{m_i}$ we want to show that
\begin{align*}
&\BU(T',\varphi,(\vec{x}_{(i)})_{i \leq n},\R{T'}) = \varphi(T')((\vec{x}_{(i)})_{i \leq n}).
\end{align*}
By its bottom-up nature, $\BU$ respects modular decomposition; more formally, 
\begin{align*}
&\BU(T',\varphi,(\vec{x}_{(i)})_{i \leq n},\R{T'}) = \BU(T,\varphi,\vec{y},\R{T}) 
\end{align*}
where $y_i = \BU(T_i,\varphi,\vec{x}_{(i)},\R{T_i})$. Since \eqref{eq:toprove} holds for $T_i$, we have $y_i = \varphi(T_i)(\vec{x}_{(i)})$; since \eqref{eq:toprove} also holds for $T$ we have
\begin{align*}
&\BU(T',\varphi,(\vec{x}_{(i)})_{i \leq n},\R{T'}) = \BU(T,\varphi,\vec{y},\R{T})\\
&= \BU(T,\varphi,(\varphi(T_i)(\vec{x}_{(i)}))_{i \leq n},\R{T})\\
&= \varphi(T)((\varphi(T_i)(\vec{x}_{(i)}))_{i \leq n}) \\
&= (\varphi(T) \star (\varphi(T_1),\ldots,\varphi(T_n)))((\vec{x}_{(i)})_{i \leq n}) \\
&= \varphi(T')((\vec{x}_{(i)})_{i \leq n}).
\end{align*}
Here we used the fact that $\varphi$ preserves $\star$ in the last line. This concludes our proof that $\BU$ works if $\varphi$ is a scoperad morphism.

To prove the opposite direction, we assume \eqref{eq:toprove} for all $T$, and we need to show that $\varphi$ preserves $\tau_{\sigma}$ and $\star$. For the first one, observe that for all $\vec{x}$ and $\sigma$ we have
\begin{align*}
\tau_{\sigma}\varphi(T)(\vec{x}) &= \varphi(T)(\sigma^*\vec{x}) \\
&= \BU(T,\varphi,\sigma^*\vec{x},\R{T}) \\
&= \BU(\tau_{\sigma}T,\varphi,\vec{x},\R{\tau_{\sigma}T}) \\
&= \varphi(\tau_{\sigma}T)(\vec{x}).
\end{align*}
For the second one, for $T'$ and $\vec{y}$ as before we have
\begin{align*}
\varphi(T')((\vec{x}_{(i)})_{i \leq n}) &= \BU(T',\varphi,(\vec{x}_{(i)})_{i \leq n},\R{T'}) \\
&= \BU(T,\varphi,\vec{y},\R{T}) \\
&= \varphi(T)(\vec{y}) \\
&= \varphi(T)((\varphi(T_i)(\vec{x}_{(i)}))_{i \leq n})\\
&= (\varphi(T) \star (\varphi(T_1),\ldots,\varphi(T_n)))((\vec{x}_{(i)})_{i \leq n}),
\end{align*}
which proves that $\varphi$ preserves $\star$. We conclude that $\varphi$ is a scoperad morphism.
\end{proof}

\section{Proof of Theorem \ref{thm:bdd}}

\begin{proof}
For $f \in \MBool_n$, let $B_f$ be the unique ROBDD representing it. We will prove by induction on $n$ the claim that $\mathtt{BDD\_BU}(B_f,\vec{x},\R{B}) = \Psi(f)(x)$ for each $f \in \MBool_n$ and each $\vec{x} \in X^n$; this proves the theorem. Since $\Psi(\tzero) = z_{\tzero}$ and $\Psi(\tone) = z_{\tone}$, it holds for the two elements of $\MBool_0$. Now assume that the claim holds for $n' < n$, and let $f \in \MBool_n$. If $t_N(\R{B_f}) \neq F_1$, then there is no $v$ with $t_N(v) = F_1$, and $f$ does not depend on its first variable; hence $f = \delta_{n-1} f'$ for some $f' \in \MBool_{n-1}$, which furthermore satisfies $B_f = B_{f'}$. It follows that
\begin{align*}
\Psi(f)(x_1,\ldots,x_n) &= \Psi(\delta_{n-1} f')(\vec{x}) \\
&= (\delta_{n-1}\Psi(f'))(\vec{x})\\
&= \Psi(f')(x_2,\ldots,x_n)\\
&= \mathtt{BDD\_BU}(B_{f'},(x_2,\ldots,x_n),\R{B_{f'}}) \\
&= \mathtt{BDD\_BU}(B_{f},x,\R{B_f}).
\end{align*}
This proves the claim for $f$. Now suppose $t_N(\R{B_f}) = F_1$. Let $B_{\tzero},B_{\tone}$ be the sub-BDDs of $B_f$ with roots $c_{\tzero}(\R{B_f}),c_{\tone}(\R{B_f})$, and let $f_{\tzero},f_{\tone}\colon \BB^{n-1} \rightarrow \BB$ be the functions they represent (where a label $F_i$ stands for the $(i-1)$-th entry of a boolean vector, as the labels now range from $F_2$ to $F_n$). Then $f_{\tzero} \prec f_{\tone}$ because $f$ is monotonous, and $f = \recht{Sh}(f_{\tzero},f_{\tone})$. Write $\breve{x} \in X^{n-1}$ for $(x_2,\ldots,x_n)$. It follows that 
\begin{align*}
& \Psi(f)(\vec{x}) \\
&= \Psi(\recht{Sh}(f_{\tzero},f_{\tone}))(\vec{x}) \\
&= g\left(x_1,\Psi(f_{\tzero})(\breve{x})),\Psi(f_{\tone})(\breve{x})\right)\\
&= g\left(x_1,\mathtt{BDD\_BU}(B_{\tzero},\breve{x},\R{B_{\tzero}}),\mathtt{BDD\_BU}(B_{\tone},\breve{x},\R{B_{\tone}})\right)\\
&= g\left(x_1,\mathtt{BDD\_BU}(B_{f},\breve{x},c_{\tzero}(\R{B_f})),\mathtt{BDD\_BU}(B_{f},\breve{x},c_{\tone}(\R{B_f}))\right)\\
&= \mathtt{BDD\_BU}(B_f,\vec{x},\R{B_f}).
\end{align*}
This proves the claim for $f$; by induction this completes the proof.
\end{proof}

\end{document}